\documentclass[aps,prx,quantum,reprint]{revtex4-2}
\pdfoutput=1
\usepackage{amsmath,amssymb,amsfonts,bm,amscd,mathtools}
\usepackage{graphicx}
\usepackage{multirow}
\usepackage{verbatim}
\usepackage{appendix}
\usepackage{fancybox}
\usepackage{array}
\usepackage{url}
\usepackage{float}
\usepackage{subfigure}
\usepackage{tikz}
\usepackage{dsfont}
\usepackage{amsthm}
\usepackage{thmtools} 
\usepackage{thm-restate}

\def\Tr{{\rm Tr \,}}

\def\CC{{\mathcal C}}

\def\CE{{\mathcal E}}

\def\CM{{\mathcal M}}

\def\CO{{\mathcal O}}
\def\CP{{\mathcal P}}

\def\CT{{\mathcal T}}

\def\CZ{{\mathcal Z}}

\def\be{\begin{equation}}
\def\ee{\end{equation}}
\def\bea{\begin{eqnarray}}
\def\eea{\end{eqnarray}}

\newcommand{\tq}{q}

\newcommand{\2}{{\mathtt 2}}

\newcommand{\locc}{{LOCC}}
\newcommand{\slocc}{{SLOCC}}

\newcommand{\1}{{1}}

\newtheorem{definition}{Definition}[section]
\newtheorem{lemma}{Lemma}[section]
\newtheorem{proposition}{Proposition}[section]

\newtheorem{remark}{Remark}[section]

\begin{document}

\title{Multi-partite entanglement monotones}

\author{Abhijit Gadde}
\affiliation{Department of Theoretical Physics \\ 
Tata Institute for Fundamental Research, Mumbai 400005}
\email{abhijit@theory.tifr.res.in}

\author{Shraiyance Jain}
\affiliation{Department of Theoretical Physics \\ 
Tata Institute for Fundamental Research, Mumbai 400005}
\email{shraiyance.jain@tifr.res.in}

\author{Harshal Kulkarni}
\affiliation{Department of Theoretical Physics \\ 
Tata Institute for Fundamental Research, Mumbai 400005}
\affiliation{Indian Institute of Science Education and Research, 
Kolkata 741246}
\email{harshalkulkarni20@gmail.com}

\begin{abstract}
If we want to transform the quantum state of a system to another using local measurement processes, what is the probability of success? This probability is bounded by quantifying entanglement in both the states. In this paper, we construct a family of local unitary invariants of multipartite states that are monotonic under local operations and classical communication on average. These monotones are constructed from local unitary invariant polynomials of the state and its conjugate, and hence are easy to compute for pure states. Using these measures we bound the success probability of transforming a given state into another state using local quantum operations and classical communication.
\end{abstract}
\maketitle

\section{Introduction and summary}\label{intro}
Let $|\psi\rangle$ be a $\tq$-partite entangled state. Assume that this state can be transformed into another state $|\phi\rangle$ using local quantum operations and classical communication (\locc)  with a non-zero probability $p_{|\psi\rangle\to |\phi\rangle}$. What more can we say about this probability? Can we bound it?
It was realized in \cite{Vidal_2000} that this question is intimately related to  quantifying entanglement. Building on the work of \cite{Vedral_1997}, the author introduced the notion of an \emph{entanglement monotone} $\mu$. It is a function of the quantum state of the system that satisfies certain properties. The most important property  is that $\mu$ does not increase  \emph{on average} under \locc. It was shown in \cite{Vidal_2000} that, 
\begin{align}
    p_{|\psi\rangle\to |\phi\rangle}\leq \frac{\mu(|\psi\rangle)}{\mu(|\phi\rangle)}.
\end{align}
Although the discussion in \cite{Vidal_2000} focused on bi-partite states, the idea is general and can be extended to multi-partite states as emphasized in \cite{10.5555/2011326.2011328}. We will review this shortly.

Let us  explain what is meant by a local quantum operation. Let us consider a mixed quantum state $\rho$ of $\tq$ number of parties $A_1, \ldots,A_q$. A local quantum operation on, say party $A_a$, transforms $\rho$ into an ensemble of mixed states $\rho_i$ each appearing with probability $p_i$. 
\begin{align}
    \rho &\xrightarrow[]{\Lambda}  \Lambda(\rho)=\{ p_i, \rho_i\},\\
    p_i &= {\rm Tr} E^{(A_a)}_i \rho E^{(A_a)\dagger}_i, \qquad \rho_i = E^{(A_a)}_i \rho E^{(A_a)\dagger}_i/p_i.\notag
\end{align}
Here  $E^{(A_a)}_i$'s are linear operators on party $A_a$ which preserve the trace of the density matrix i.e. $\sum_i E^{(A_a)\dagger}_i E^{(A_a)}_i={\mathbb I}$. They are known as the Kraus operators. The local quantum operation on party $A_a$ can also be thought of as a two step process. First, a joint unitary operation on party $A_a$ and its environment $R$ and then a projective measurement on $R$. The measurement outcome $i$ occurs with probability $p_i$ and yielding the state $\rho_i$. 

If one thinks of entanglement as a resource, then it must not increase under a local quantum operation. This property is formalized as $\mu_{\rm avg}(\Lambda(\rho)) \leq \mu (\rho)$ where  $\mu_{\rm avg}$ is the weighted average of $\mu$ on the resulting ensemble i.e. $\mu_{\rm avg} (\Lambda(\rho))= \sum_i p_i \mu (\rho_i)$. It follows that $\mu$ must be invariant under local unitary transformation because such a transformation  is an invertible local quantum operation and  that $\mu$ is monotonic under local quantum operations. 

The individual parties can choose to dismiss the information about measurement outcomes and treat the ensemble $\Lambda(\rho)$ as a mixed state  with density matrix $\Lambda(\rho)_{\rm mixed}= \sum_i p_i \rho_i$. The entanglement monotone is expected to be non-increasing under such operation. This is modeled in \cite{Vidal_2000} as the convexity condition. 
\begin{align}
    \sum_i p_i \mu(\rho_i) \geq \mu (\sum_i p_i \rho_i).
\end{align}
It is pointed out in \cite{Plenio_2005} that the monotonicity under dismissal of information does not correspond to convexity but rather to  monotonicity under partial trace.   Neither of these conditions play an important role for us.

Before we formalize the definition of entanglement monotones  we need to introduce the notion of fully separable states. 
\begin{definition}[Fully separable state]
    A $\tq$-partite mixed state $\rho$ is called fully separable if it can be written as
    \begin{align}\label{sep}
        \rho =\sum_i p_i \, \rho_{A_1}^{(i)}\otimes \ldots \otimes \rho_{A_q}^{(i)}.
    \end{align}
    where $p_i>0$ and $\rho_{A_a}^{(i)},a=1,\ldots, q$ are local density matrices.
\end{definition}
\begin{definition}[Entanglement monotone]\label{em}
    A local unitary invariant function $\mu$ of the density matrix $\rho$ is called an entanglement monotone if the following conditions hold.
    \begin{enumerate}
        \item \rm{(monotonicity under unilocal operations)} 
        For all parties $A_a$, it should hold that 
        \begin{align}\label{locc-mono}
            &\mu_{\rm avg}(\Lambda(\rho)) \leq \mu (\rho), \quad
            \mu_{\rm avg} (\Lambda(\rho)):=\sum_i p_i \mu (\rho_i),\notag\\
            &\rho \xrightarrow[]{\Lambda}  \Lambda(\rho)=\{ p_i, \rho_i\},\notag\\
            &p_i = {\rm Tr} E^{(A_a)}_i \rho E^{(A_a)\dagger}_i, \, \rho_i = E^{(A_a)}_i \rho E^{(A_a)\dagger}_i/p_i.
        \end{align}
         $E^{(A_a)}_i$'s are linear transformations on party $A_a$ which preserve trace i.e. $\sum_i E^{(A_a)\dagger}_i E^{(A_a)}_i={\mathbb I}$.
        \item {\rm (convexity)} If $\rho=\sum_i p_i \rho_i$ then
        \begin{align}\label{convex}
            \mu(\rho) \leq \sum_i p_i \mu (\rho_i).
        \end{align}
        \item \rm{(faithfulness)}
        $\mu(\rho)=0$ if $\rho$ is fully separable. 
    \end{enumerate}
\end{definition}
\noindent
Note the distinction between the first two conditions. At first go, the first condition may look like a concavity condition and hence contradicting the second condition of convexity, but that is not so. The first condition is a statement about the value of the monotone \emph{after} the \locc\ operation $\Lambda$ while the second one is simply a convexity condition. In a sense, the first condition is more physical while the second is more mathematical. Also note that the first and the second condition together imply monotonicity under ``deterministic'' operations i.e.  
\begin{align}
    \mu(\Lambda(\rho)_{\rm mixed}) \leq \mu(\rho),\quad \Lambda(\rho)_{\rm mixed} := \sum_i p_i \rho_i
\end{align}
Here by $\Lambda(\rho)_{\rm mixed}$ we mean the mixed state corresponding to the ensemble $\{p_i,\rho_i\}$ where $p_i$ and $\rho_i$ are as defined in equation \eqref{locc-mono}.  The third condition implies $\mu(\rho)\geq 0$ because \locc\ operations on a fully separable state keep the state fully separable. 

As remarked earlier, not all \locc\ monotones are convex. Logarithmic negativity \cite{Plenio_2005} is one such example. Although we stick to the above definition of monotones for our review exposition, the convexity condition does not play any role in the derivation of our main results which are about restriction of entanglement monotones to pure states. 

The bi-partite entanglement monotones that agree with entanglement entropy on pure states are often called entanglement measures. This is due to the operational meaning of entanglement entropy as entanglement cost and also as distillable entanglement for bi-partite pure states. See \cite{Plenio:1998zr, schumacher2000relative, 10.5555/2011326.2011328, horodecki2009distillation, Plenio:2007zz} for a review of entanglement measures. 
For multi-partite states, understanding entanglement from an operational point of view is a difficult task. We only consider the axiomatic approach of quantifying entanglement using entanglement monotones defined above. 

\begin{restatable}[]{rem}{}
    If $\mu_i$'s are entanglement monotones then their convex combination $\sum_i p_i \mu_i$ is also an entanglement monotone. In other words, the set of entanglement monotones is convex. 
\end{restatable} 
\noindent In fact, entanglement monotones satisfy a stronger property \cite{Szalay_2015}.
\begin{restatable}[Szalay\cite{Szalay_2015}]{rem}{szalay}\label{szalay}
    If $G(x_1,\ldots, x_k)$ is a concave and monotonic function of its arguments such that $G(0,\ldots, 0)=0$, then $G(\nu_1, \ldots, \nu_k)$ is an entanglement monotone if all $\nu_i$s are entanglement monotones. 
\end{restatable}
\noindent The proof is reproduced in appendix \ref{review-proofs}.

The relevance of entanglement monotones to the problem of bounding the transition probability $p_{\rho\to \rho_*}$ under \locc\ can be seen thanks to the following theorem. We have included in the proof in appendix \ref{review-proofs} for the reader's convenience.  

\begin{restatable}[Vidal \cite{Vidal_2000}]{thm}{vidala}\label{bound-vidal}
    The maximal success probability $p_{\rho\to \rho_*}$ of converting a multipartite state $\rho$ to $\rho_*$ using \locc\ is
    \begin{align}\label{opt}
        p_{\rho\to \rho_*}=\min_{\mu}\frac{\mu(\rho)}{\mu(\rho_*)}.
    \end{align}
    where the minimization is performed over all entanglement monotones $\mu$.
\end{restatable}
\noindent

Constructing entanglement monotone directly from its  definition \ref{em} is difficult. Instead it is convenient to first construct what we call   \emph{pure state entanglement monotone} ({PSEM}) $\nu$. As the name suggests, it is defined only for pure states. The  definition is given in \ref{pure-em}. Note that this notion  is \emph{a priori} independent of entanglement monotones defined in \ref{em}. However a  {PSEM} $
\nu$ can be extended to mixed states using \emph{convex roof extension} yielding a full-fledged entanglement monotone ${\tt CR}_\nu(\rho)$. This is reviewed in appendix \ref{convex-roof}. 
This shows that even though the {PSEM} is  defined independently it ends up being restriction of an  entanglement monotone to pure states. The advantage of formalizing the notion of {PSEM} independently is that if we are interested in transition probability $p_{|\psi\rangle\to |\phi\rangle}$, where $|\psi\rangle$ and $|\phi\rangle$ are pure states then we need only concern ourselves with the {PSEM}s and not their convex roof extension to mixed states. Moreover we can still use the theorem \ref{bound-vidal} for {PSEM} $\nu(|\psi\rangle)$, thanks to the fact that it is restriction of an  entanglement monotone to pure states. What makes working with {PSEM}s even more convenient is their characterization in terms of concavity. 
\begin{restatable}[Vidal \cite{Vidal_2000}]{thm}{vidalb}\label{concave}
    A local unitary invariant function $f(|\psi\rangle)$ that is concave in the partial trace $\rho_{\bar{A_a}}:= {\rm Tr}_{A_a} |\psi\rangle\langle \psi|$ for all parties $A_a, a=1,\ldots, q$ is a {PSEM}.
\end{restatable}
\noindent 
A local unitary invariant function $f$ of the $q$-partite pure state $|\psi\rangle$ can equivalently be thought of as a function of $\rho_{\bar{A_a}}$ for any $a=1,\ldots, q$. It is defined by purifying $\rho_{\bar{A_a}}$ to $|\tilde \psi\rangle$ and evaluating the function $f(|\tilde \psi\rangle)$. Because all the purifications are related by local unitary on the purifier and because the function $f$ is local unitary invariant, $f(|\tilde \psi\rangle)=f(|\psi\rangle)$. More on this in section \ref{psem}.   
The proof is reviewed in appendix \ref{review-proofs}.
In this paper, we will construct {PSEM}s for multipartite states and the concavity theorem \ref{concave} will be our main workhorse for doing so.  

Although, the question of entanglement monotones for bi-partite states has been explored widely in the literature, the same for multi-partite states has received relatively less attention. Some multi-partite entanglement monotones are known in the literature (see \cite{Horodecki_2009, Szalay_2015} for a review) but a general theory is lacking. Our aim in this paper is to remedy this situation and develop a theory of multi-partite monotones of a certain type. 
Let us summarize the main results of the paper. In this paper we consider local unitary invariant polynomials of the state and its conjugate. Such a polynomial invariant was  termed \emph{multi-invariant} in \cite{Gadde:2024taa}. We stick to this nomenclature. As we will discuss in section \ref{new}, $\tq$-partite multi-invariants are characterized by uniformly edge-labeled bi-partite graph i.e. a bi-partite graph whose every vertex neighborhood consists of $\tq$ edges with the same set of distinct $\tq$ labels. We call such graphs $\psi$-graphs. Such graphs occur naturally in group theory. Specifically, Cayley graph of a group with $\tq$ involutive generators is a uniformly edge-labeled graph with $\tq$ edge-labels, each edge-label corresponding to a generator. If all the relation are even-length words, then such a graph is also bi-partite and hence a $\psi$-graph. 
Of course not all $\psi$-graphs are Cayley graphs. We will denote the $\psi$-graph as well as associated multi-invariant with a calligrphic letter such as $\CZ$. A normalized multi-invariant ${\hat \CZ}$ is defined as $\CZ^{1/n_\CZ}$ where $n_\CZ$ is the number of black (or white) vertices in the $\psi$-graph $\CZ$. 

\subsection{Main results}
Motivated by theorem  \ref{szalay}, we have
\begin{definition}\label{primitive}
    If a {PSEM} can be written as $G(\nu_1, \ldots, \nu_k)$ where $G$ is a concave and monotonic function of {PSEM}s $\nu_i$ then we call it a composite of $\nu_i$s otherwise we call it a primitive {PSEM}.
\end{definition}
\noindent 
We prove the following general theorem about entanglement monotones 
that shows composites do not give better bounds on transition probability than the associated primitives.
\begin{restatable}{thm}{composite}\label{composite}
    If $G(x_1,\ldots, x_k)$ is a positive, monotonic and concave function of its arguments $x_i\geq 0$ in some domain then in that domain,
    \begin{align}\label{comp}
        \frac{G(x_1,\ldots, x_k)}{G(x'_1,\ldots, x'_k)} \geq \min \Big(
            \frac{x_1}{x'_1},\ldots, \frac{x_k}{x'_k},1
            \Big).
    \end{align}
\end{restatable}
\noindent 
The proof is presented in appendix \ref{proof-theorem1}. 
This theorem is useful because it makes it clear that we do not need to explore the bounds on transition probabilities obtained from composites of {PSEM}s that we construct. We prove another general result that constructs new {PSEM} from a given {PSEM} in conjunction with theorem \ref{concave}. 
\begin{restatable}{thm}{generalproj}\label{general-proj}
    Let $f(\rho_{\bar{A}})$ be a linearly homogeneous convex function of $\rho_{\bar{A}}$ that is invariant under local unitary transformations. Then the function 
    \begin{align}
        \tilde f_{k_1,\ldots k_\tq}(|\psi\rangle) :=\max_{P_{k_1,\ldots, k_\tq}} f(P_{k_1,\ldots, k_\tq} |\psi\rangle)
    \end{align}
    also is a linearly homogeneous convex function of $\rho_{\bar{A}}$ that is invariant under local unitary transformations.
\end{restatable}
\noindent 
We use the notation $P_{k_1,\ldots, k_\tq}=P_{k_1}\otimes \ldots \otimes P_{k_\tq}$ where $P_{k_{a}}$ is a $k_{a}$ dimensional projector acting on party $A_a$. The action of a projector on a normalized state $|\psi\rangle$ results in an  unnormalized state. The function $f(|\psi\rangle)$ is extended to such states using its linear homogeneity i.e. $f(\alpha|\psi\rangle)=\alpha f(|\psi\rangle)$.
The proof of this theorem is presented in appendix \ref{proof-theorem1}. This is a standalone result, it hasn't been explored further in this paper.

The rest of our results are about the construction of {PSEM}s from multi-invariants. Let us define $\nu (\CZ):=1-\CZ$ and ${\hat \nu}(\CZ) := 1-{\hat \CZ}$. Then we have,
\begin{restatable}{thm}{theorempsem}\label{theorem-psem}
    ${\hat \nu}(\CZ)$ is a {PSEM} if $\CZ$ is connected and  edge-convex.
\end{restatable}
\noindent The proof is given in section \ref{strengthen}. 
Note that if $\hat \nu(\CZ)$ is a {PSEM} then so is $\nu(\CZ)$ because the latter is a composite of the former. Edge-convexity is a  property of $\psi$-graphs that  is defined in \ref{edge-convex-def}. 

If $\CZ_i$ is a $\tq_i$-partite $\psi$-graph then we define an operation -- similar to the cartesian product of graphs -- that we call the box product $\CZ_1{\hat \Box}\CZ_2$. This box product  is a $\tq_1+\tq_2$-partite $\psi$-graph. We prove
\begin{restatable}{thm}{cartesian}\label{cartesian}
    If $\CZ_1$ and $\CZ_2$ are both edge-convex then $\CZ_1{\hat \Box}\CZ_2$ is edge-convex.
\end{restatable}
\noindent
This gives us a simple way to construct edge-convex $\psi$-graphs. It turns out straightforwardly that ``1-partite'' $\psi$-graph ${\cal E}_1$ corresponding to the norm of the state is edge-convex. We also show explicitly that the bi-partite $\psi$-graph $C_n$ corresponding to ${\rm Tr}\rho^n$ is edge-convex. Combining these results with theorem \ref{cartesian}, we conclude that the $\psi$-graph
\begin{align}
    {\cal E}_{m;n_1, \ldots, n_k}:={\cal E}_1^{{\hat \Box} m} {\hat \Box} C_{n_1}{\hat \Box} \ldots {\hat \Box} C_{n_k}
\end{align}
is edge-convex. The number of parties in this $\psi$-graph is $m+2(n_1+\ldots+n_k)$. The logarithm of ${\hat {\cal E}}_{m; \emptyset}$ is the so called $n=2$ Renyi multi-entropy studied in \cite{Gadde:2022cqi}. Also, the analytic continuation of case ${\cal E}_{1;n}$ to $n=1/2$ yields what is known as computable cross-norm.  

Let us emphasis the broader significance of our results. Multipartite entanglement is central to distributed quantum information processing, yet the space of entanglement monotones beyond the bipartite setting remains poorly charted. This work provides a systematic route to constructing multipartite pure-state entanglement monotones with direct operational meaning: each monotone yields an explicit upper bound on the success probability of stochastic LOCC conversion between pure states (when such a conversion is possible). Our construction starts from multiplicative local-unitary invariant polynomials, called multi-invariants and introduces a graph-theoretic  condition - edge-convexity - that guarantees the associated normalized invariant produces a monotone. The resulting family is algebraically explicit (polynomial in state amplitudes and their conjugates), naturally composable via the box-product, and connects entanglement theory to the combinatorics of edge-labeled bipartite graphs underlying LU invariants.

The multi-invariants are expressed as expectation values of permutation operators acting on multiple copies,
$Z_{\vec\sigma}=\mathrm{Tr}\!\left[\rho^{\otimes k}U_{\vec\sigma}\right]$,
where $U_{\vec\sigma}$ is a tensor product of party-wise permutations across the $k$ copies.
Consequently, they are, in principle, estimable without full tomography using standard multi-copy measurement primitives.
A direct route is \emph{ancilla-assisted interferometry}  implementing a controlled-$U_{\vec\sigma}$,
in which measurements of the control in the $X/Y$ bases yield
$\mathrm{Re}\,\mathrm{Tr}(\rho^{\otimes k}U_{\vec\sigma})$ and
$\mathrm{Im}\,\mathrm{Tr}(\rho^{\otimes k}U_{\vec\sigma})$,
generalizing the controlled-SWAP network for nonlinear state functionals~\cite{Ekert:2002qtj, LeiferLindenWinter2004MeasuringPolynomialInvariants}.
Related multi-copy \emph{interference protocols} have been realized experimentally to access SWAP-type expectation values and R\'enyi entropies~\cite{Islam2015MeasuringEntanglementEntropy}.
Alternatively, when controlled permutations are experimentally demanding, \emph{randomized measurement} schemes based on local unitary designs reconstruct such multi-copy moments from correlations of single-copy measurement outcomes~\cite{Elben2018RenyiRandomQuenches,Brydges2019ProbingRenyiRandomized};
complementary ``classical shadows'' methods provide a broader framework for predicting many nonlinear properties from few randomized measurements~\cite{Huang2020ClassicalShadows}.
Since our entanglement monotones are explicit functions of normalized multi-invariants, these tools provide practical ways to estimating the monotones directly from state.

The rest of the paper is organized as follows. In  section \ref{psem}, we will define pure state entanglement monotone ({PSEM}) and characterize it as a concave function of partially traced density matrix. We will also review the known bi-partite and multi-partite {PSEM}s.
Section \ref{new} contains our new results about multi-partite {PSEM}s, constructed from multi-invariants.
We first develop a graph theoretic language to deal with local unitary invariant functions of the state and formulate the condition of concavity in theorem \ref{concave} as the ``edge-convexity''  condition on the associated graph. We explicitly show the edge-convexity of ${\cal E}_{2;\emptyset}, {\cal E}_{3;\emptyset}$ and $\CC_n$. 
Then we present the theorem \ref{cartesian}. 
Finally,  we demonstrate the effectiveness of these new monotones with the help of an example. In the last section, section \ref{sec:conclusions}, we give conclusions and list some future research directions.

\section{Pure state entanglement monotones}\label{psem}
In this section, we will define {PSEM}  and   describe its  characterization  as a concave function of partially traced density matrix. Then we will take stock of the known bi-partite and multi-partite {PSEM}s culminating in the statement of theorem \ref{general-proj}.

Entanglement monotones are difficult to construct in general. A strategy to construct monotones for bi-partite systems was given in \cite{Vidal_2000}. This strategy can be generalized to multi-partite systems as well. It is reviewed in \cite{10.5555/2011326.2011328}.  The idea is that one starts with a function $\nu(|\psi\rangle)$ that is defined only on \emph{pure states} and then extend it to mixed states using convex roof (see appendix \ref{convex-roof}) to get an entanglement monotone. 
\begin{definition}[Pure state entanglement monotone ({PSEM})]\label{pure-em}
    Let $|\psi\rangle \to \Lambda(|\psi\rangle)=\{p_i, |\psi_i\rangle\}$. A pure state entanglement monotone  $\nu(|\psi\rangle)$ is defined as obeying the following properties for all parties $A_a$,
    \begin{align}\label{locc-mono-pure}
        &\nu (|\psi\rangle) \geq \sum_i p_i \nu (|\psi_i\rangle), \\
        &{p_i}:=|E^{(A_a)}_i |\psi\rangle|^2, \quad|\psi_i\rangle := E^{(A_a)}_i |\psi\rangle /\sqrt{p_i},\notag
    \end{align}
    where $E^{(A_a)}_i$ act on party $A_a$ and obey the trace preserving condition $\sum_i E_i^{(A_a)\dagger}E_i^{(A_a)\dagger}={\mathbb I}$.  $\nu(|\psi\rangle)=0$ for fully separable states.
\end{definition}
\noindent Fully separable pure states are fully factorized i.e. $|\psi\rangle=|\psi_{A_1}\rangle \otimes\ldots\otimes  |\psi_{A_q}\rangle$.
The fact that all {PSEM}s can be extended to entanglement monotones allows for the use of {PSEM}s to bound pure state transition probabilities. In particular, the optimal success probability $p_{|\psi\rangle\to |\phi\rangle}$ of converting a multi-partite pure state $|\psi\rangle$ to another multi-partite pure state $|\phi\rangle$ using \locc\ is 
\begin{align}
    p_{|\psi\rangle\to |\phi\rangle} = \min_{\nu} \frac{\nu(|\psi\rangle)}{\nu(|\phi\rangle)}. 
\end{align} 
where the minimization is performed over the set of all {PSEM}s.

\begin{definition}\label{complete}
    A set ${\cal K}$ of {PSEM}s is called complete if
    \begin{align}
        p_{|\psi\rangle\to |\phi\rangle}=\min_{\nu\in {\cal K}}\frac{\nu(|\psi\rangle)}{\nu(|\phi\rangle)} \qquad \qquad \forall\,\, |\psi\rangle, |\phi\rangle.
    \end{align}
    A set that is a subset of all complete sets is called the minimal complete set.
\end{definition}
\noindent
It is clear that all the primitive {PSEM}s form a complete set and that the {PSEM}s in a minimal complete set must be primitive. 
For bi-partite pure states, Vidal computed the optimal probability by constructing a \emph{minimal complete} set of {PSEM}s. We will review this in section \ref{bi-psem}. Characterizing the minimal set for multi-partite pure states is a difficult exercise, one which we will not undertake in this paper.

Vidal gave a simple and elegant way of constructing {PSEM}s. 
\vidalb*
\noindent 
As remarked earlier, because of local unitary invariance, $f(|\psi\rangle)$ can be equivalently thought of as a function of the partial trace  $\rho_{\bar{A_a}}$ for any party $A_a$. 
Let us see why that is the case in detail.
Given a $\rho_{\bar{A_a}}$ on $\tq-1$ parties, we purify it to a pure state $|\tilde \psi\rangle$ on $\tq$ parties by essentially reintroducing the traced out party $A_a$. Let us denote the new party as $\tilde A_a$. The dimension of $H_{\tilde A_a}$ needs to be at least as large as the rank of $\rho_{\bar{A_a}}$ but it can be larger.   It is well-known that all such purifications are related to each other by a local unitary transformation on $\tilde A_a$. Now we define $f(\rho_{\bar{A_a}})=f(|\tilde \psi\rangle)$ and because $f(|\psi\rangle)$ is a local unitary invariant function, $f(\rho_{\bar{A_a}})$ doesn't depend on the purification and hence is uniquely defined.
By abuse of notation, we will denote such a function either as $f(|\psi\rangle)$ or as $f(\rho_{\bar{A_a}})$, according to convenience.

The main point of the paper is to construct {PSEM}s for multi-partite states using the theorem \ref{concave}.

\subsection{For bi-partite states}\label{bi-psem}
Because the local unitary invariant data in the partial trace $\rho_{\bar{A}}$ of a bi-partite state are its eigenvalues $\lambda_i$ and because $\lambda_i$ are homogeneous in $\rho_{\bar{A}}$, {PSEM}s are constructed using concave functions of $\lambda_i$, thanks to theorem \ref{concave}. Also note that the eigenvalues $\lambda_i$ are symmetric with respect to both parties as required by the theorem \ref{concave}.
In \cite{Vidal_1999}, Vidal identified a minimal complete set of {PSEM}s that is finite. Arranging the nonzero eigenvalues of $\rho_{\bar{A}}$ in descending order, $\lambda_1\geq \lambda_2 \geq \ldots$, define
\begin{align}
    \tilde \nu_k^{V} = 1-\sum_{i=1}^k \lambda_i.
\end{align}
It turns out that 
\begin{restatable}[Vidal \cite{Vidal_1999}]{thm}{vidalc}\label{kyfan}
    $\tilde \nu_k^V$ are {PSEM}s.
\end{restatable}
\noindent
The proof is reviewed in appendix \ref{review-proofs}. 
Moreover, the set ${\cal K}=\{\tilde \nu_k^{V}: k=1,\ldots, d\}$, where $d$ is the number of nonzero eigenvalues of $\rho_{\bar{A}}$, is complete. Vidal proved the completeness of ${\cal K}$ by giving an explicit  conversion protocol $|\psi\rangle \to |\phi\rangle $ that realizes the success probability obtained by minimizing over ${\cal K}$. 

Note that when 
\begin{align}\label{maj}
    \min_k \,\frac{\tilde \nu_k^V(|\psi\rangle)}{\tilde \nu_k^V(|\phi\rangle)} \geq 1,
\end{align}
Vidal's theorem states that it is possible to convert $|\psi\rangle$ to $|\phi\rangle$ with unit probability. This agrees beautifully with Nielsen's result \cite{Nielsen:1999zza} that a complete conversion of $|\psi\rangle$ to $|\phi\rangle$ is possible using \locc\ if and only if eigenvalues of $\rho_{\bar{A}}(|\psi)$ majorize eigenvalues of $\rho_{\bar{A}}(|\phi\rangle)$. The condition \eqref{maj} is simply a rewriting of Nielsen's majorization condition.  

It is useful to write an expression for $\tilde \nu_k^V$ that is manifestly symmetric in both parties. For this, we note that
\begin{align}
    \max_{P_{k_1,k_2}} |P_{k_1,k_2}|\psi\rangle|^2 = \sum_{i=1}^{{\rm min}(k_1,k_2)} \lambda_i.
\end{align}
Here $P_{k_1,k_2}:= P_{k_1}\otimes P_{k_2}$ where $P_{k_1}$ and $P_{k_2}$ are the rank $k_1$ and rank $k_2$ projectors on both the parties respectively. So we define another form of the same monotones,
\begin{align}\label{kyfan2}
    \nu_{k_1,k_2}^V:= 1- \max_{P_{k_1,k_2}} |P_{k_1,k_2}|\psi\rangle|^2=\tilde \nu_{{\rm min}(k_1,k_2)}^V.
\end{align}
Here we have introduced redundancy by writing the monotone as a function of two integers $k_1$ and $k_2$ although it is only a function of ${\rm min}(k_1,k_2)$. This symmetric form allows generalization to higher parties as shown in \cite{Barnum_2001}. We discuss their multi-partite {PSEM} shortly. 

\subsection{For multi-partite states}\label{multi}
Before we start constructing {PSEM}s for multi-partite states, let us look at the problem of conversion $|\psi\rangle \to |\phi\rangle$. In general, this conversion is possible only if $|\psi\rangle\langle\psi|$ is related to $|\phi\rangle\langle\phi|$ by a stochastic version of local operation and classical communication (\slocc) \cite{D_r_2000}. Mathematically an \slocc\ is expressed as,
\begin{align}\label{slocc}
    &|\psi\rangle\to |\phi\rangle := M_{A_1}\otimes \ldots \otimes M_{A_\tq} \, |\psi\rangle|_{\rm normalized},\notag\\
    &M_{A_a}^{\dagger} M_{A_a} \preceq {\mathbb I}\quad {\rm for}\quad a=1,\ldots, q.
\end{align}
The notation $M\preceq {\mathbb I}$ implies that the matrix ${\mathbb I}-M$ is positive definite. 
Indeed this has to be the case if $|\phi\rangle$ were to appear as one of the possible mixed states under local operations. 
The reason for the condition $M_{A_a}^{\dagger} M_{A_a} \preceq {\mathbb I}$ is that $M_{A_a}$ belongs to a set of Kraus operators, in this case a pair, the other being ${\tilde M}_{A_a}:=\sqrt {{\mathbb I}-M_{A_a}^{\dagger} M_{A_a}}$. 
In the bi-partite case, any two density matrices (of equal rank) are related to each other in this way but that is not always the case for multi-partite density matrices or even for multi-partite pure states. If they are not in the same \slocc\ orbit then the success probability of converting either $|\psi\rangle \to |\phi\rangle$ or $|\phi\rangle \to |\psi\rangle $ is zero. Hence the problem of bounding conversion success only makes sense if $|\psi\rangle$ and $|\phi\rangle$ are in the same \slocc\ orbit i.e. if $|\phi\rangle= M_{A_1}\otimes \ldots \otimes M_{A_\tq} \, |\psi\rangle$. In the rest of the paper, we will assume that that is the case. Whether given two states belong to the same \slocc\ orbit can be checked by computing \slocc\ invariant quantities. Such quantities are constructed by contracting the indices of $|\psi\rangle$ using only the invariant $\epsilon$ tensor for each party. 

This discussion also suggests a particular protocol of transformation $|\psi\to |\phi\rangle$ thereby giving a lower bound on the transformation probability. The protocol consists of the action of a pair of Kraus operators  $M_{A_a},{\tilde M}_{A_a}$ each party $A_a$. In this case, $|\phi\rangle \langle\phi|$ definitely appears as one of the $2^{\tq}$ terms in the ensemble (it may also appear elsewhere in the sum as well). Its coefficient is 
\begin{align}\label{lowerb}
    p_*={\rm Tr}M_{A_1}\otimes \ldots \otimes M_{A_\tq} \, |\psi\rangle \langle\psi| \, M_{A_a}^{\dagger}\otimes \ldots \otimes M_{A_\tq}^{\dagger}. 
\end{align}
This gives us a lower bound on $p_{|\psi\rangle \to |\phi\rangle}$. 
In order to get the best lower bound for this protocol, we take all the $M_{A_a}$'s to be normalized such that their maximum singular value is $1$. This allows the inequality $M_{A_a}^{\dagger} M_{A_a} \preceq {\mathbb I}$ to be satisfied while giving the best lower bound for this protocol. Of course, there might exist other protocols which allow for a more efficient conversion of $|\psi\rangle $ to $|\phi\rangle$. 

Now we give some examples of {PSEM}s for multi-partite states. 

\subsection*{{PSEM} of \cite{PhysRevA.68.012103}}
In \cite{PhysRevA.68.012103}, the authors showed that any linearly  homogeneous positive function of $\rho=|\psi\rangle\langle \psi|$ that is invariant under local special linear  transformations i.e. under transformation
\begin{align}
    &|\psi\rangle\to |\phi\rangle := M_{A_1}\otimes \ldots \otimes M_{A_\tq} \, |\psi\rangle|_{\rm normalized},\notag\\
    &{\rm with}\qquad {\rm Det} \,M_{A_a}=1 \quad{\rm for}\quad a=1, \ldots, q.
\end{align}
is a {PSEM}. Because of this invariance, for a given \slocc\ orbit the expression for the monotone simplifies to
\begin{align}\label{det-psem}
    \nu_{\rm Det}(|\psi\rangle) = \kappa\,  \langle \psi| \psi\rangle^{-1}
\end{align}
where $|\psi\rangle$ is an un-normalized state obtained from some reference state $|\psi_0\rangle$ in the same \slocc\ orbit by determinant $1$ operations and $\kappa$ is a constant that only depends on the orbit.

\subsection*{{PSEM} of \cite{PhysRevA.64.022306}}
In \cite{PhysRevA.64.022306} the authors defined the Schmidt measure ({\tt SM}) of a multi-partite state,
\begin{definition}[Schmidt measure]
    {\tt SM} is the $\log_2$ of the minimum value of $k$ such that 
    \begin{align}
        |\psi\rangle = \sum_{i=1}^k c_i |\psi\rangle_{A_1}^{(i)}\otimes \ldots \otimes |\psi\rangle_{A_q}^{(i)}.
    \end{align}
\end{definition}
\noindent They showed that {\tt SM} is a {PSEM}. 

\subsection*{{PSEM} of \cite{Barnum_2001}}
A natural generalization of Vidal's monotones for multi-partite case was constructed in \cite{Barnum_2001} using the Ky Fan's maximum formulation as presented in equation \eqref{kyfan2}. Consider
\begin{align}\label{bl}
    &\nu_{k_1,\ldots, k_\tq}^{BL}:= 1-\max_{P_{k_1,\ldots,k_\tq}} |P_{k_1,\ldots,k_\tq}|\psi\rangle|^2,\notag\\
    & P_{k_1,\ldots,k_\tq}=P_{k_1}\otimes \ldots \otimes P_{k_\tq}.
\end{align}
\noindent
It was shown in \cite{Barnum_2001} that $\nu_{k_1,\ldots, k_\tq}^{BL}$ are multi-partite {PSEM}s. 

We will now state a powerful theorem which can be specialized to prove that $\nu^{BL}$ are {PSEM}s. 
Given a convex function of partial traces, we will construct a new convex function of partial traces. This function can be used to define a {PSEM}. 
\generalproj*
\noindent The proof is given in appendix \ref{proof-theorem1}. 
If we take $f(\rho_{\bar{A}}) = {\rm Tr} \rho_{\bar{A}}$. Then the above theorem shows that the measure \eqref{bl} introduced in \cite{Barnum_2001} is indeed a {PSEM}.

Unfortunately almost all of the existing  multi-partite {PSEM}s involve optimization and are  difficult to compute. 
In the rest of the paper, we will construct  multi-partite {PSEM}s that are easier to compute. Interestingly, they will be constructed from convex functions of $\rho_{\bar{A}}$ that are linearly homogeneous, hence they are amenable to the generalization discussed in theorem \ref{general-proj}.

\section{{PSEM}s from polynomial invariants}\label{new}
In this section, we will construct a set of {PSEM}s for multi-partite states from local unitary invariant polynomials of the wave-function and its complex conjugate. In particular, we will construct the local unitary invariants that are convex in partially traced density matrix $\rho_{\bar{A}}$ for each party $A$. They will lead us to {PSEM}s thanks to theorem \ref{concave}. Local unitary invariant functions of a multi-partite quantum states have been explored widely. See \cite{vrana2011alge, Szalay_2012, Vrana_2012, Vrana_2011, hero2011measure, HERO20096508, Rains2000PolynomialInvariantsQuantumCodes, BrylinskiBrylinski2002InvariantPolyQuditsChapter, GrasslRoettelerBeth1998ComputingLocalInvariants} for their enumeration and construction. 

Let $|i_a\rangle, \, i_a=\1,\ldots, d_{A_a}$ be a basis for Hilbert space $H_{A_a}$. A state in the tensor product $\bigotimes_{a=1}^q H_{A_a}$ is written as 
\begin{align}
    |\psi\rangle = \sum \, \psi_{i_\1,\ldots, i_\tq}\, |i_\1\rangle \otimes \ldots\otimes |i_\tq\rangle.
\end{align}
The components $\psi_{i_\1,\ldots, i_\tq}$ is the wavefunction of the state $|\psi\rangle$ in the chosen basis. The index $i_a$ transforms in the fundamental representation of the unitary group acting on party $A_a$. The conjugate wavefunction is ${\bar\psi}^{i_\1,\ldots, i_\tq}$. Its indices transform in the anti-fundamental representation. 
Invariants of local unitary transformations are constructed by taking, say $n_r$  copies of $\psi$ and $n_r$ copies of $\bar \psi$ and contracting the fundamental indices of $\psi$'s with the anti-fundamental indices of $\bar \psi$'s as dictated by permutation elements of $S_{n_r}$ associated with each party. More explicitly,
\begin{align}\label{general-multi}
    \CZ(|\psi\rangle)= \langle\psi^{\otimes n_r}| (\sigma_{A_1}\otimes \ldots\otimes\sigma_{A_q}) |\psi^{\otimes n_r}\rangle
\end{align}
where $\sigma_{A_a}\in S_{n_r}$ is a permutation operator acting on $n_r$ copies of party $A_a$ and so on. The local unitary invariant thus defined is called multi-invariant \cite{Gadde:2024taa}. As this function is defined as the  expectation value of a unitary operator on multiple copies of the state, it can even be measured using the protocol of \cite{Ekert:2002qtj, LeiferLindenWinter2004MeasuringPolynomialInvariants}. In what follows below, we will sometimes label the parties as $A,B,C,\ldots$ rather than $A_1,A_2,A_3,\ldots$ to avoid clutter. The index for party $A$ is then denoted as $i_A$.

It is convenient to use a graphical notation to describe multi-invariants. 
Let us denote a state $\psi_{i_A,i_B,i_C,\ldots}$ (its complex conjugate $\bar\psi_{i_A,i_B, i_C,\ldots}$) as a white (black) $\tq$-valent vertex. 
Each edge has a label of one of the $\tq$-parties. The edge corresponding to party $A$ is called an $A$-edge and so on. This notation is illustrated in figure \ref{psi-notation}. 
In the figures, we denote the edge label using a color that is not black or white. 
\begin{figure}[h]
    \begin{center}
        \includegraphics[scale=0.25]{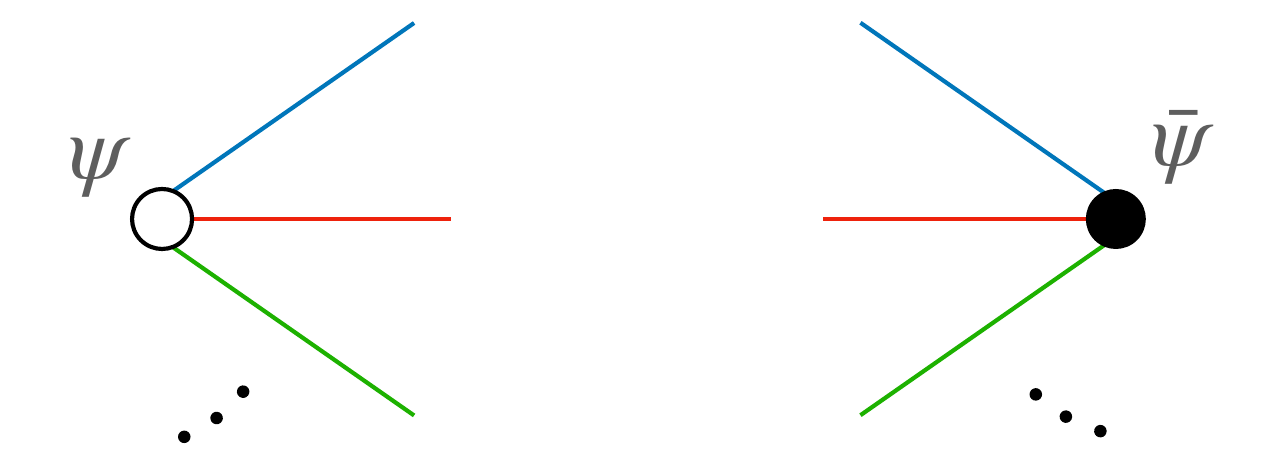}
    \end{center}
    \caption{White (black) vertex denoting $\psi$ ($\bar \psi$). The parties are labeled by colored edges. }\label{psi-notation}
\end{figure}
Whenever an index $i_A$ of a pair of $\psi$ and $\bar \psi$ is contracted, we connect the two corresponding vertex with $A$-edge and so on. If all the edges are contracted, the graph represents a local unitary invariant and if some edges are left unconnected then the open graph represents a tensor that transforms appropriately under local unitary transformations as indicated by the uncontracted indices. The multi-invariant in equation \eqref{general-multi} is obtained by connecting $A$-edge of $\alpha$-th white vertex to $(\sigma_{A}\cdot \alpha)$-th black vertex for all $A$. In this way, a multi-invariant is given by a bi-partite, $\tq$-color-regular graph. We call such a graphical representation of the multi-invariant a  $\psi$-\emph{graph}. Figure \ref{example} shows an example of a $\psi$-graph. 
\begin{figure}[h]
    \begin{center}
        \vspace{0.5cm}
        \includegraphics[scale=0.3]{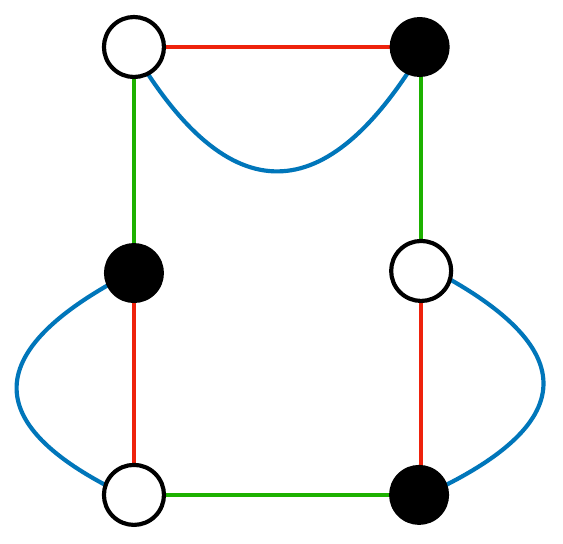}
    \end{center}
    \caption{Example of a $\psi$-graph constructed from three copies of $\psi$ and $\bar \psi$ each by connecting edges of identical colors.}\label{example}
\end{figure}
We use the calligraphic letter, such as $\CZ$, denoting the multi-invariant to also denote the associated $\psi$-graph as well. Multi-invariants obeys the factorization property
\begin{remark}\label{factor}
    \begin{align}
        \CZ(|\psi_1\rangle\otimes |\psi_2\rangle)= \CZ(|\psi_1\rangle)\CZ(|\psi_2\rangle).
    \end{align}
    Here $|\psi_1\rangle$ and $|\psi_2\rangle$ are $\tq$-partite states and their tensor product is also thought of as a $\tq$-partite state with  party $A$ of $|\psi_1\rangle\otimes |\psi_2\rangle$ being the tensor product of  party $A$ in $|\psi_1\rangle$ and  party $A$ in $|\psi_2\rangle$.
\end{remark} 
\begin{remark}\label{smaller}
    For a normalized state $|\psi\rangle$,
    \begin{align}
        |\CZ(|\psi\rangle)|\leq 1,
    \end{align} 
    with equality holding for fully factorized state. 
\end{remark}
\noindent 
This is because the tensor product of permutation operators $\sigma_A\otimes \sigma_B\ldots$ is a unitary operator. Then  remark \ref{smaller} holds due to Cauchy-Schwarz inequality. 

For future convenience, it is useful to define normalized multi-invariant $\hat \CZ:=\CZ^{1/n_r}$ where $n_r$ is  the number of white (or black) vertices in the graph $\CZ$. Obviously remarks \ref{factor} and \ref{smaller} are also valid for the normalized multi-invariants.  

In what follows, we will discuss a special class of multi-invariants called symmetric multi-invariants. We will relate their $\psi$-graphs to certain types of Cayley graphs. For this purpose, it is convenient to think of the un-oriented colored edges as a pair of oppositely oriented edges.

\subsubsection{Reflection symmetry and positivity}
In this section we will discuss the properties of $\psi$-graph that is reflection symmetric.  
Let us first recall that a graph cut of a connected graph is a subset of edges after removing which, the graph becomes disconnected. From now on, without loss of generality, we will assume $\CZ$ is connected. The reason for this is explained below theorem \ref{theorem-psem}. 
\begin{definition}
    An edge subset $E$ of a $\psi$-graph $\CZ$ is called a reflecting cut if there exists an automorphism $k$ that flips vertex color such that
    \begin{enumerate}
        \item For every edge $uv\in E$, $k(v)=u$ and $k(u)=v$.
        \item $\CZ-E$ consists of two components $\CT_1$ and $\CT_2$ that are disconnected from each other and are mapped to each other by $k$.
    \end{enumerate}
\end{definition}
\noindent 
The $\psi$-graph in figure \ref{example} admits a reflecting cut. It is shown in figure \ref{example-cut}.
\begin{figure}[h]
    \begin{center}
        \includegraphics[scale=0.3]{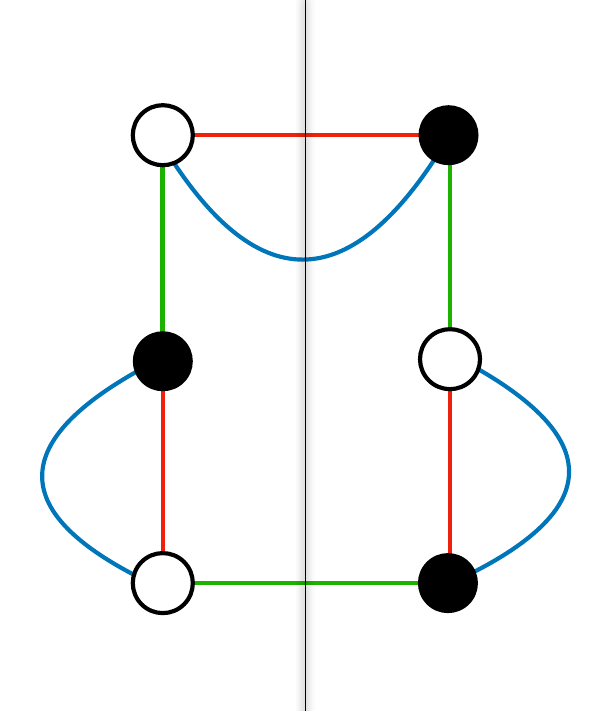}
    \end{center}
    \caption{The reflecting cut is shown by a straight black line passing through the graph. It is easy to see that the graph is symmetric under the reflection across the reflecting cut, after vertex color flip.}\label{example-cut}
\end{figure} 

\begin{remark}\label{reflecting-cut}
    If a $\psi$-graph $\CZ$ admits a reflecting cut then $
    \CZ$ is positive. 
\end{remark}
\noindent This can be seen as follows. Consider the two graphs obtained after the reflecting cut. Restoring the cut edges on each of them separately but not joining them gives us a pair of open graphs that represents tensors $|\CT_1\rangle$ and $| \CT_2\rangle$  that are complex conjugates of each other i.e. $| \CT_2\rangle=| {\bar \CT}_1\rangle$. The original $\psi$-graph $\CZ$ is obtained by connecting the  indices that are images of each other. This gives the presentation of $\CZ$ as the squared norm of $|\CT_1\rangle$. Hence $\CZ$ is positive.

\subsection{Strategy}\label{strategy}

Now let us outline our strategy of constructing {PSEM}s from multi-invariants $\CZ$. Any multi-invariant is invariant under local unitary transformations and as explained below theorem \ref{concave}, any local unitary invariant function of $|\psi\rangle$ can be thought of as a function of $\rho_{\bar A}:={\rm Tr}_A|\psi\rangle\langle \psi|$. We will construct $\CZ$ that is convex in $\rho_{\bar A}$ for all parties $A$. Thanks to theorem \ref{concave} and remark \ref{smaller}, $1-\CZ$ is then a {PSEM}. We will take $\CZ$ to have at least one relfecting cut. As a result $\CZ>0$.

Let us  comment on how to think of $\CZ$ as a function of $\rho_{\bar A}$ graphically. The density matrix $\rho_{\bar{A}}$ is constructed by joining a black and a white vertex only by $A$-edge and keeping all the other edges open as shown in figure \ref{rho-notation}. 
\begin{figure}[h]
    \begin{center}
        \includegraphics[scale=0.25]{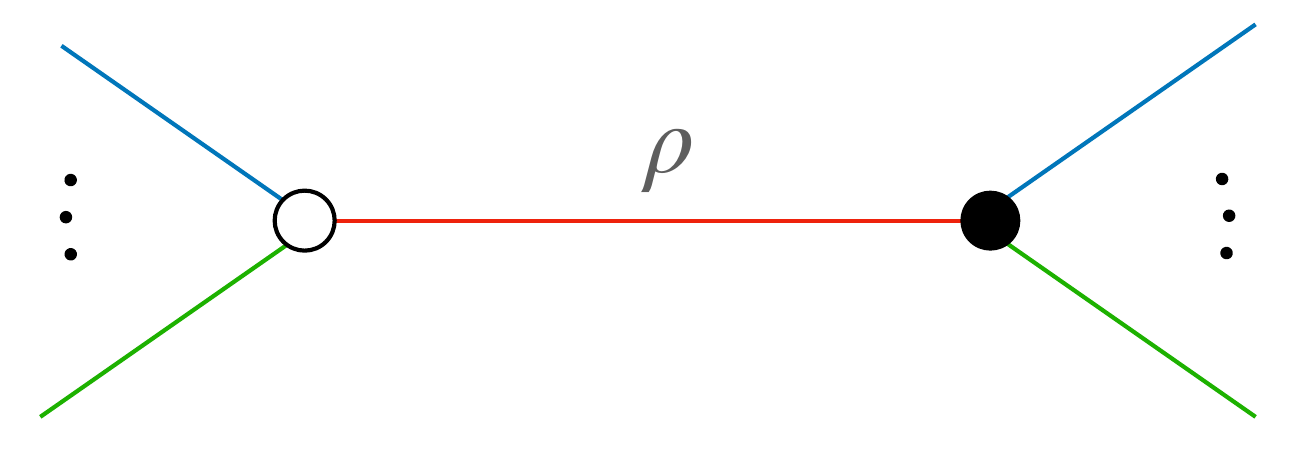}
        \includegraphics[scale=0.25]{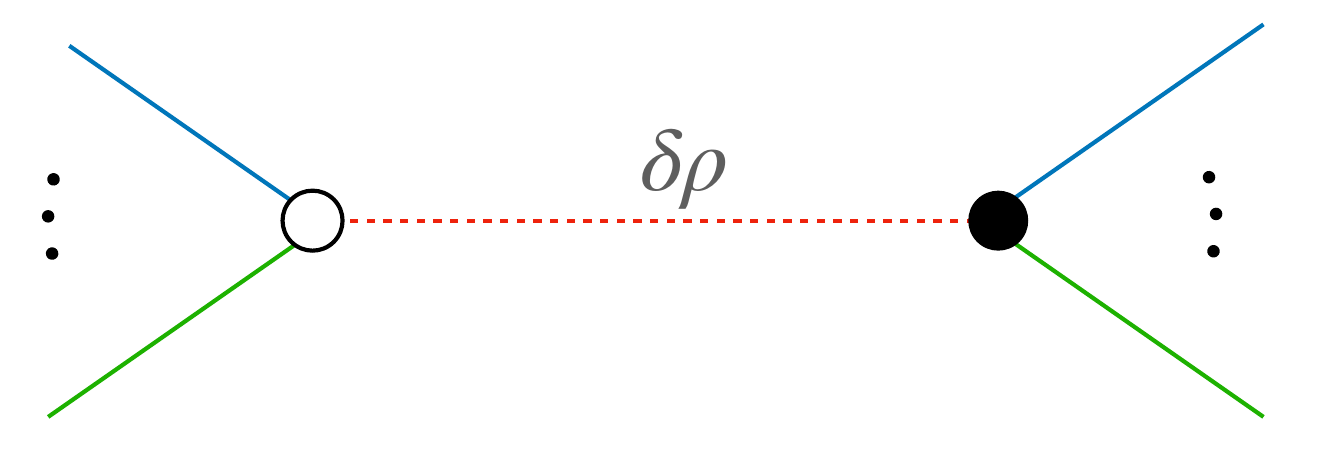}
    \end{center}
    \caption{Connecting only the $A$-edge (denoted by red color) of a $\psi$-$\bar \psi$ pair to get $\rho_{\bar{A}}$. The dotted line represents insertion of $\delta\rho_{\bar{A}}$ instead of $\rho_{\bar{A}}$. }\label{rho-notation}
\end{figure}
In this sense $\rho_{\bar{A}}$ ``lives on'' the $A$-edge. In order to analyze convexity of $\CZ$, we need compute derivative of $\CZ$ with respect to $\rho_{\bar A}$. We do so by varying $\rho_{\bar A}\to \rho_{\bar A}+\epsilon \delta \rho_{\bar A}$ for some  hermitian matrix $\delta \rho_{\bar A}$. The graphical expression for derivatives involves replacing  $\rho_{\bar{A}}$ on a certain edge $e$ of the graph by the variation  $\delta \rho_{\bar A}$.  
Graphically, we denote $\delta\rho_{\bar{A}}$ in the same way as $\rho_{\bar{A}}$ but with a dotted line. We will call the dotted line representing $\delta\rho_{\bar{A}}$ as a $\delta A$-edge. This is also shown in figure \ref{rho-notation}.

Now we are ready to compute the variation of $\CZ(\rho_{\bar A})$ under $\rho_{\bar{A}}\to \rho_{\bar{A}}+\epsilon \delta\rho_{\bar{A}}$ where $\delta\rho_{\bar{A}}$ is a hermitian matrix. Expanding $\CZ(\rho_{\bar{A}}+\epsilon \delta\rho_{\bar{A}})$ in $\epsilon$,
\begin{align}
    \CZ(\rho_{\bar{A}}+\epsilon \delta \rho_{\bar{A}}) = \CZ(\rho_{\bar{A}})+ \epsilon \delta\CZ+\frac{1}{2}\epsilon^2 \delta^2 \CZ+\CO(\epsilon^3).
\end{align} 
Here $\delta\CZ$ and $\delta^2\CZ$ is the first and second derivative of $\CZ$ with respect to $\rho_{\bar A}$. The convexity is of $\CZ$ with respect to $A$ is equivalent to,
\begin{align}\label{strat-ineq}
    \delta^2\CZ &\geq 0.
\end{align}
for all $\delta \rho_{\bar A}$.

We will ensure the convexity of $\CZ$ in $\rho_{\bar A}$ by expressing $\delta^2\CZ$ as sum of norms of vectors. The multi-invariant $\CZ$ for which this is possible is called \emph{$A$-edge-convex}. We will formalize the definition later in  definition \ref{edge-convex-def}. If $\CZ$ is $A$-edge-convex for all $A$, then it is called \emph{edge-convex}. 
The condition of $A$-edge-convexity is sufficient for convexity of $\CZ$ with respect to $\rho_{\bar A}$. It is not clear to us if it is necessary as well. It would be interesting to explore this point further.

For edge-convex $\psi$-graphs, we will prove a result result that is stronger than convexity of $\CZ$. We will show that if $\CZ$ is a connected graph and is $A$-edge-convex then the normalized multi-invariant, $\hat \CZ:=\CZ^{1/n_r}$ is also a convex function $\rho_{\bar A}$. Again, along with remark \ref{smaller}, this implies that for a connected $\CZ$ that is edge-convex, 
\begin{align}
    \hat \nu:=1-\hat \CZ
\end{align} 
is a {PSEM}. The {PSEM} $\hat \nu$  is stronger than $\nu$ because  $\nu=f(\hat \nu)= 1-(1-\hat \nu)^{n_r}$. The function $f$ is a concave monotonic function of $\hat \nu$ in the range $[0,1]$ for $n_r>1$. Because $0\leq {\CZ}\leq 1$, $\hat \nu$ is indeed in the range $[0,1]$. This makes $\nu$ a composite of $\hat \nu$ and hence the probability bounds following from it are not stronger than those following from $\hat \nu$, thanks to theorem \ref{composite}. 

In order to analyze the inequality \eqref{strat-ineq} and express $\delta^2\CZ$ as sum of norms,  it is useful to understand the derivatives graphically.
\begin{align}
    \delta\CZ=\sum_e \delta_e\CZ,\qquad {\rm and}\qquad \delta^2 \CZ=\sum_{e',e}^{e\neq e'} \delta_e\delta_{e'}\CZ
\end{align}
Here $\delta_e\CZ$ is an invariant obtained from $\CZ$ by replacing $\rho_{\bar{A}}$ at $A$-edge $e$ by $\delta\rho_{\bar{A}}$ i.e. by replacing the $A$-edge $e$ by $\delta A$-edge and $\delta_e\delta_{e'}\CZ$  is an invariant obtained by replacing $A$-edges $e$ and $e'$ by $\delta A$-edges. The quantity $\delta_e\delta_{e'}\CZ$ is symmetric in $e,e'$. 
For the example of $\psi$-graph in figure \ref{example}, we have shown the graphical representation of  $\delta\CZ$ and $\delta^2 \CZ$ below in figure \ref{delta-z}.
\begin{figure}[h]
    \begin{center}
        \includegraphics[scale=0.15]{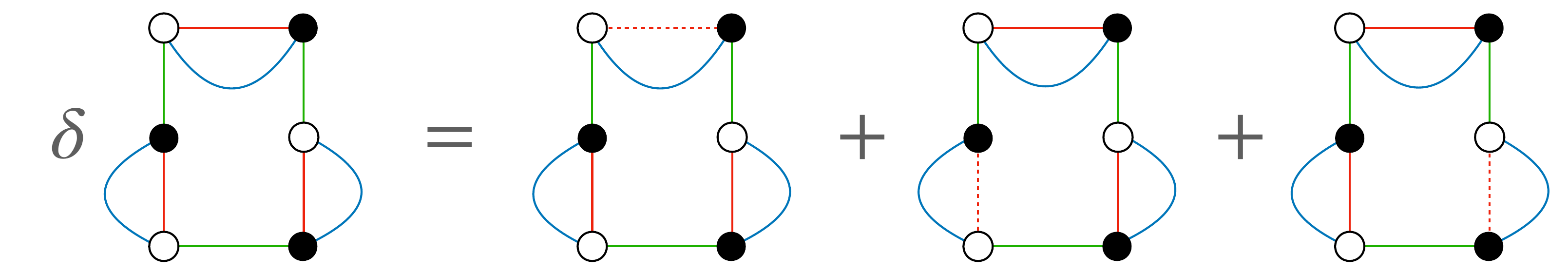}\\
        \includegraphics[scale=0.15]{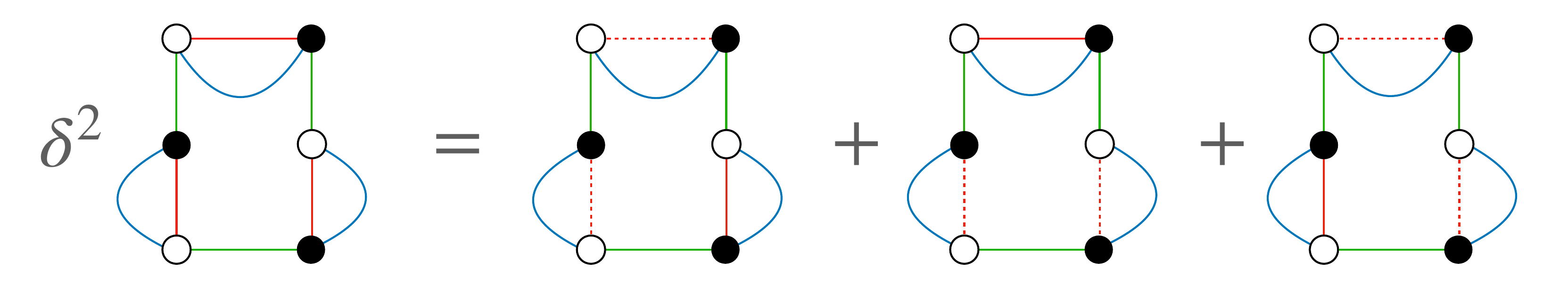}
    \end{center}
    \caption{The first line denotes $\delta \CZ$ for $\CZ$ given in figure \ref{example}. Each of the $A$-edges has been replaced by $\delta\rho_{\bar{A}}$ and the terms are added up. The second line denotes $\delta^2 \CZ/2$ for same $\CZ$. A pair of $A$-edges have been replaced by $\delta\rho_{\bar{A}}$ and the terms are added up.}\label{delta-z}
\end{figure}
The first inequality in \eqref{strat-ineq} is proved by expressing the left hand side $\delta^2\CZ$  as a sum of squared norms i.e. 
\begin{align}\label{strat1}
    \delta^2\CZ= \sum_i \langle \CT_i|\CT_i\rangle
\end{align}
for some states $|\CT_i\rangle$. To prove convexity of $\hat \CZ$ for an edge-convex $\CZ$, we will show
\begin{align}\label{strat2}
    &\CZ \delta^2\CZ - (1-\frac{1}{n_r})|\delta \CZ|^2\notag\\
    &= \sum_i \langle \CT| \otimes \langle \CT_i| ({\mathbb I}-{\mathbb P}) |\CT\rangle \otimes |\CT_i\rangle
\end{align}
Here $|\CT\rangle \otimes |\CT_i\rangle$ is a product state. The operator ${\mathbb I}$ is identity and the operator ${\mathbb P}$ is a permutation operator that swaps the two factors of the product state. As ${\mathbb P}$ is a unitary operator with eigenvalues $\pm 1$, the right hand side of \eqref{strat2} is non-negative. Interestingly, the set of states $|\CT_i\rangle$ in equation \eqref{strat2} ends up being  the same as the one that appears in equation \eqref{strat1}. That is why it is useful to analyze the first inequality in \eqref{strat-ineq}.

\subsection{Some simple examples}\label{simple-examples}
In this section we will demonstrate the strategy outlined in section \ref{strategy} in two relatively simple examples of a bi-partite and tri-partite {PSEM}. We will then consider a general multi-invariant and formulate convexity graphically using the strategy of section \ref{strategy}.

\subsubsection{Example 1: Bi-partite}
Consider a bi-partite state $\psi_{A\bar A}$. Tracing over party $A$ we obtain $\rho_{\bar A}$. Consider the multi-invariant $\CZ(|\psi\rangle)=\CE^{(2)}(|\psi\rangle):={\rm Tr}\rho_{\bar A}^2$. It is clearly symmetric in $A$ and $\bar A$. So in order to construct {PSEM} we only need to show convexity of $\CZ$ and $\hat \CZ$ in  $\rho_{\bar A}$. 

In case of $\CZ=\CE^{(2)}$, $\CZ$ consists of only two $A$-edges, say $e_1$ and $e_2$. As a result $\delta^2 \CZ$ consists of a two identical terms $\delta_{e_1}\delta_{e_2}\CZ$ and $\delta_{e_2}\delta_{e_1}\CZ$. Each of these terms is ${\rm Tr}\delta\rho_{\bar A}^2$.  Its positivity is seen as follows, 
\begin{align}
    \delta^2 \CZ=2{\rm Tr}\delta\rho_{\bar A}^2=\langle\sqrt{2}\delta \rho_{\bar A}|\sqrt{2}\delta \rho_{\bar A}\rangle.
\end{align}
In this case, the set of states $|\CT_i\rangle$ consists of a single element, namely $|\sqrt{2}\delta\rho_{\bar A}\rangle$. The argument is presented graphically in figure \ref{arg1}.
\begin{figure}[h]
    \begin{center}
        \includegraphics[scale=0.25]{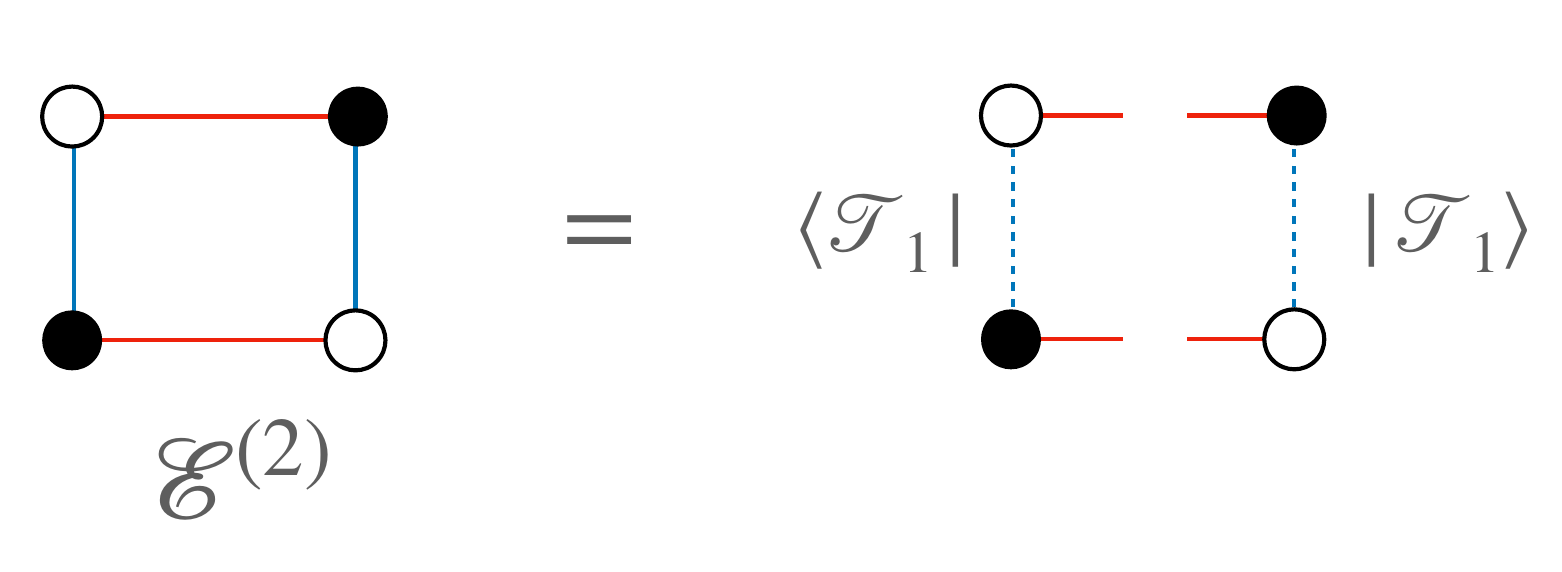}
    \end{center}
    \caption{Expressing $\delta^2{\CE}^{(2)}$ as norm of state $|\CT_1\rangle$.}\label{arg1}
\end{figure}

To prove convexity of $\hat \CZ= \CZ^{1/2}$, we simply notice the positivity of the expectation value
\begin{align}
    \langle \rho_{\bar A}\otimes \sqrt{2}\delta\rho_{\bar A}|({\mathbb I}-{\mathbb P})|\rho_{\bar A}\otimes \sqrt{2}\delta\rho_{\bar A}\rangle 
\end{align}
as outlined in section \ref{strategy}. As before, ${\mathbb P}$ is a permutation operator that swaps the two factors of the product state.
Because $\delta_e \CZ=\delta \CZ/2$, the above quantity is the same as $\CZ\delta^2 \CZ-\frac12 |\delta \CZ|^2$. This shows that $\hat \nu :=1-\hat \CZ$ is a {PSEM}. This argument is presented graphically in figure \ref{arg2}.
\begin{figure}[h]
    \begin{center}
        \includegraphics[scale=0.25]{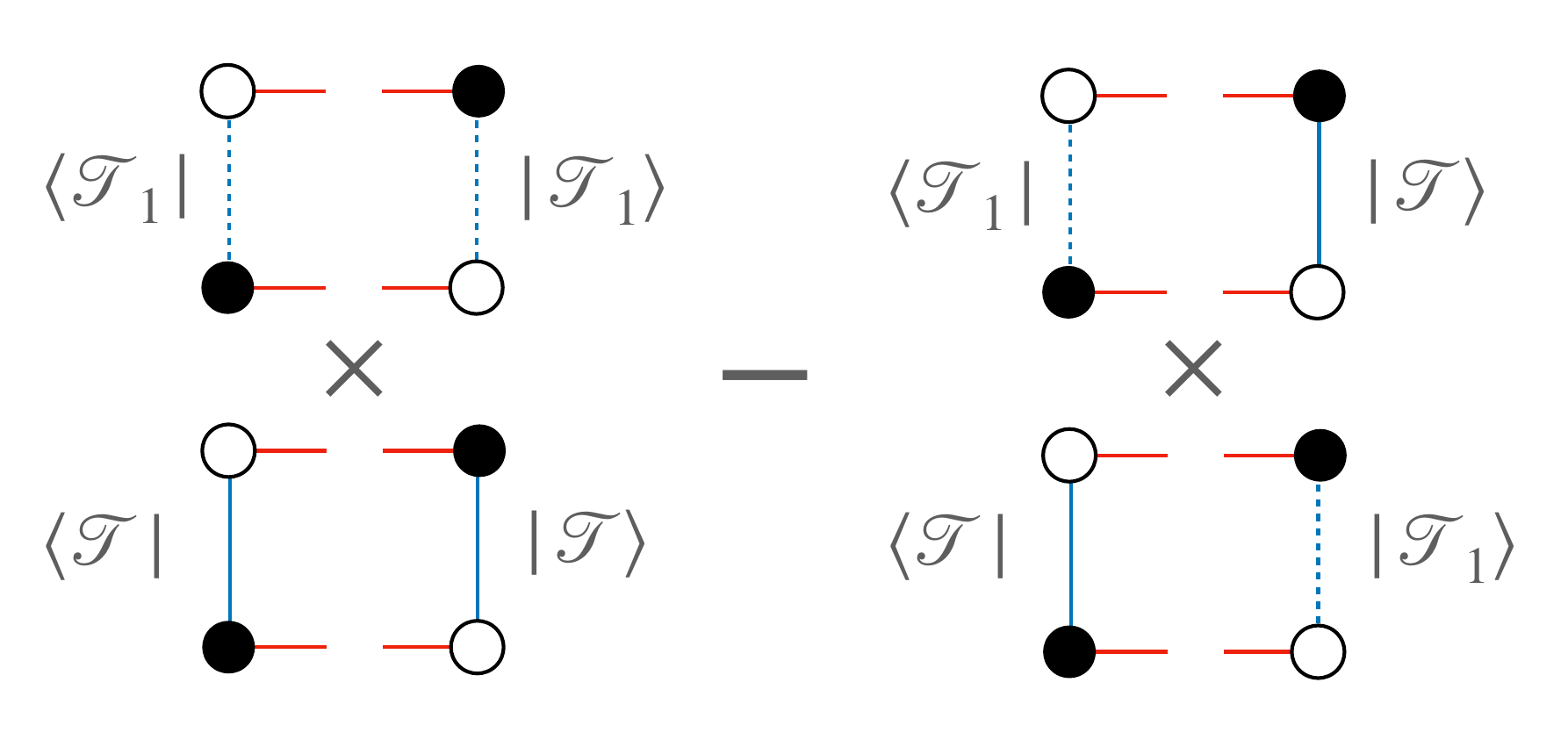}
    \end{center}
    \caption{Expressing $\CZ\delta^2 \CZ-\frac12 |\delta \CZ|^2$ as expectation value of a positive semi-definite operator ${\mathbb I}-{\mathbb P}$ in state $|\CT\rangle \otimes |\CT_1\rangle$.}\label{arg2}
\end{figure} 
The conclusion of this analysis is that $1-\sqrt{{\rm Tr}\rho_{\bar A}^2}$ is a {PSEM}.

\subsubsection{Example 2: Tri-partite}
We will show that the tripartite invariant $\CE^{(3)}$ defined in figure \ref{e3pic} and its normalized version $\hat \CE^{(3)}$ is convex in $\rho_{\bar A}$ for all $A$. 
\begin{figure}[h]
    \begin{center}
        \includegraphics[scale=0.25]{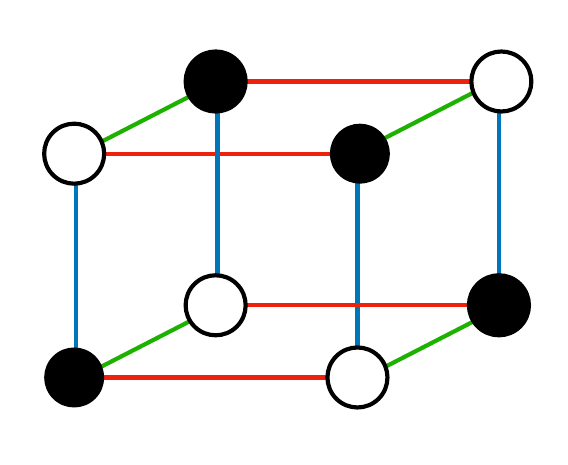}
    \end{center}
    \caption{The above $\psi$-graph defines a tri-partite multi-invariant ${\cal E}^{(3)}$.}\label{e3pic}
\end{figure} 
As the invariant is symmetric in all three parties, we only need to show convexity with respect of $\rho_{\bar A}$ for one of the parties $A$. In order to do so we simply need to produce a set of states $|\CT_i\rangle$ which obeys the equations \eqref{strat1} and \eqref{strat2}. Unlike the bi-partite case, this set consists three states. We have expressed 
\begin{align}
    \delta^2{\CE^{(3)}} =\langle {\cal T}_1|{\cal T}_1\rangle+\langle {\cal T}_2|{\cal T}_2\rangle+ \langle {\cal T}_3|{\cal T}_3\rangle
\end{align}
graphically in figure \ref{e3norm}. The figure also verifies the equations \eqref{strat1} and \eqref{strat2}.
\begin{figure}[h]
    \begin{center}
        \includegraphics[scale=0.15]{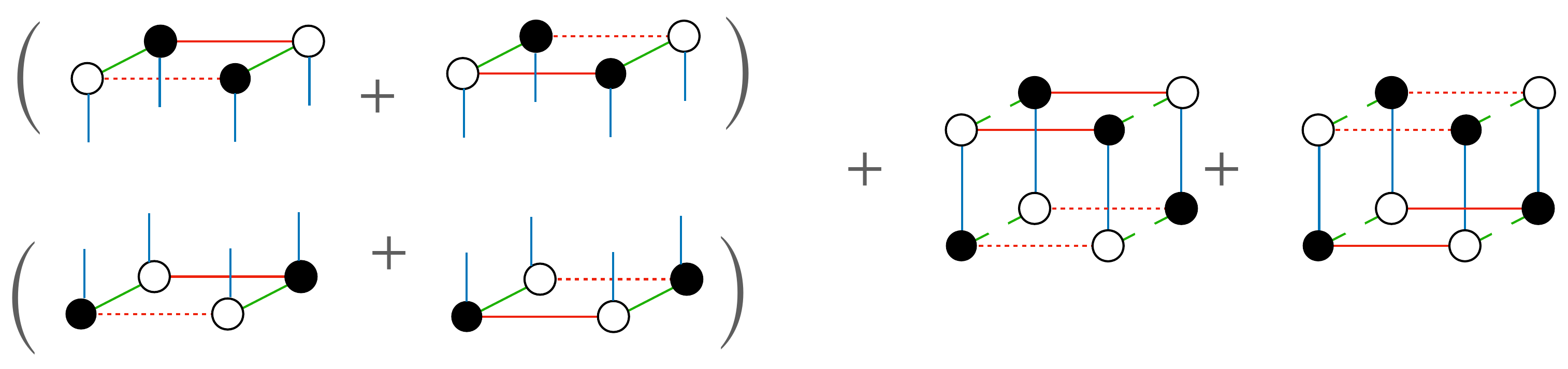}
    \end{center}
    \caption{Expressing $\delta^2{\CE^{(3)}}$ as the sum of norms of two states $|\CT_1\rangle, |\CT_2\rangle $ and $|\CT_3\rangle$. The first state $|\CT_1\rangle$ is sum of two terms.}\label{e3norm}
\end{figure}

\subsection{Edge-convexity}\label{constraints}
As exemplified by the bi-partite and tri-partite invariants, reflecting cuts play a crucial role in establishing the inequalities \eqref{strat-ineq}. We will consider a general multi-invariant $\CZ$ that admits multiple reflecting cuts. For a given reflecting cut $k$ of $\CZ$, let us denote the open graph on the left side of the cut as $L_k$ and the one on the right side of the cut as $R_k$.  
The open graph $R_k$ defines the state $|\CT^{(k)}\rangle$. The open graph $L_k$  defines the conjugate state $\langle\CT^{(k)}|$.
Let us also define the state obtained from $|\CT^{(k)}\rangle$ by replacing $\rho_{\bar{A}}$ at the $A$-edge $e$ by $\delta\rho$ as $|\CT_{e}^{(k)}\rangle$. Observe that 
\begin{align}
    &\langle \CT^{(k)}|\CT_{e}^{(k)}\rangle=\delta_{e} \CZ,\notag\\
    &\langle \CT_{e}^{(k)}|\CT^{(k)}\rangle =\delta_{k(e)} \CZ,\notag\\
    &\langle \CT_{e'}^{(k)}|\CT^{(k)}_{e}\rangle=\delta_{k(e')}\delta_{e} \CZ.
\end{align}
Here $k(e)$ denotes the image of $e$ under reflection defined by the cut $k$. Note that $k(k(e))=e$. The inner product is taken by gluing the open edges that are images of each other under reflection.

We now consider multiple states of the type $\sum_{e \in R_k} a_e |\CT^{(k)}_e\rangle$ for various choices of coefficients and for all reflecting cuts $k$ and add their norms to get the positive quantity,
\begin{align}
    &\sum_k \Big(|\sum_{e\in R_k} a_e^{(k)}|\CT^{(k)}_e \rangle |^2 +|\sum_{e\in R_k}  b_e^{(k)}|\CT^{(k)}_e\rangle |^2+\ldots\Big) >0\notag\\
    &\sum_k \sum_{ e,e'\in R_k} \Big( a_{e'}^{(k)*}a_e^{(k)}+b_{e'}^{(k)*}b_e^{(k)} +\ldots \Big) \, \delta_{k(e')}\delta_e\CZ  >0\notag\\
    &\sum_k \sum_{ e,e'\in R_k} \CP_{e,e'} ^{(k)} \, \delta_{k(e')}\delta_e\CZ >0.
\end{align} 
Here $\CP_{e,e'} ^{(k)}:= a_{e'}^{(k)*}a_e^{(k)}+b_{e'}^{(k)*}b_e^{(k)} +\ldots$ is a positive definite matrix. It is indexed by $A$-edges on the right side of the reflecting cut $k$. For $e'\in R_k$, defining $k(e')=e''\in L_k$, we have
\begin{align}
    \sum_k \sum_{e\in R_k, e'\in L_k} \CM_{e,e'} ^{(k)} \, \delta_e\delta_{e'}\CZ &>0.
\end{align}  
Here $\CM^{(k)}_{e,e'} :=\CP^{(k)}_{e,k(e')}$ where $e$ and $e'$ are on the opposite sides of the reflecting cut. To prove convexity of $\CZ$, we need to be able to choose positive definite matrices $\CP^{(k)}$ such that the left-hand side is $\sum_{e\neq e'} \delta_e\delta_{e'}\CZ=\delta^2\CZ$. Each ordered term $(e,e')$ with $e\neq e'$ must appear precisely once.
\begin{align}\label{sol-convex}
    \sum_{k\, {\rm s.t.}\,e\in R_k, e'\in L_k} \,\,\CM_{e,e'}^{(k)}=1. \qquad \forall\,  e,e'.
\end{align}
where the sum $\sum_{k}$ is carried out over reflecting cuts $k$ which separate $e$ and $e'$. 
\begin{definition}\label{edge-convex-def}
    If a $\psi$-graph admits a solution to equation \eqref{sol-convex} then it is called $A$-edge-convex. If it is $A$-edge-convex for all parties then it is called edge-convex. 
\end{definition}
\noindent
We have shown,
\begin{proposition}\label{main1}
    If a $\psi$-graph $\CZ$ is edge-convex then $\nu:=1-\CZ$ is a {PSEM}.
\end{proposition}

Now we extract a simple and necessary condition for a $\psi$-graph to be edge-convex.
\begin{definition}
    A $\psi$-graph is called $A$-edge-reflecting if it admits a reflecting cut that separates any pair of $A$-edges $(e,e')$. If it is $A$-edge-reflecting for all $A$ then it is called edge-reflecting. 
\end{definition}
\begin{remark}\label{obvious1}
    If a $\psi$-graph is A-edge-convex then it is $A$-edge-reflecting.
\end{remark}
\noindent
This is because, for a given pair of vertices $e,e'$, there needs to be at least one reflecting cut separating them so that it has a chance of appearing on the left-hand side sum in equation \eqref{sol-convex} which is the defining  condition for edge-convexity. 

\begin{remark}\label{transitive}
    If a connected $\psi$-graph is  A-edge-reflecting then it is $A$-edge-transitive.
\end{remark}
\noindent
Here by $A$-edge transitivity we mean that the automorphism group of the $\psi$-graph acts transitively on $A$-edges.  Let $e_1$ and $e_2$ be a pair of $A$-edges. Find a path between $e_1$ and $e_2$ that has alternating $A$-edges. Reflecting cuts containing a non-A-edge in this path maps consecutive $A$-edges to each other. Sequence of such reflecting cuts yields an automorphism that maps $e_1$ to $e_2$.

\subsubsection{A look back at the examples}
It is instructive to express the proof of convexity of $\CE^{(2)}$ and $\CE^{(3)}$ in terms of the positive definite matrices $\CP^{(k)}$ associated to each reflecting cut $k$. 

The $\psi$-graph ${\cal E}^{(2)}$ is a square as shown in figure \ref{cut2}. We have also shown the reflecting cut $k$ separating $A$-edges (denoted by red color). The associated positive matrix is $\CP^{(k)}=(1)$.
\begin{figure}[h]
    \begin{center}
        \includegraphics[scale=0.3]{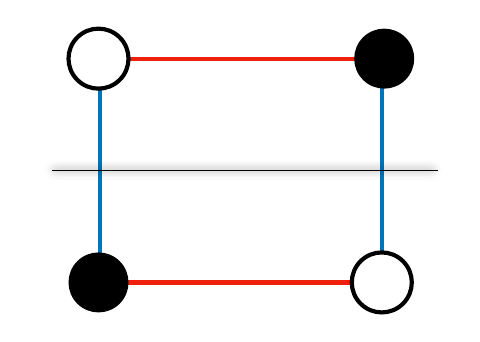}
    \end{center}
    \caption{Graph of ${\cal E}^{(2)}$. The black line denotes the reflecting cut $k$}\label{cut2}
\end{figure} 

The $\psi$-graph ${\cal E}^{(3)}$ is a cube as shown in figure \ref{cut3}, along with the two reflecting cuts $k_1$ and $k_2$ separating $A$-edges.
\begin{figure}[h]
    \begin{center}
        \includegraphics[scale=0.25]{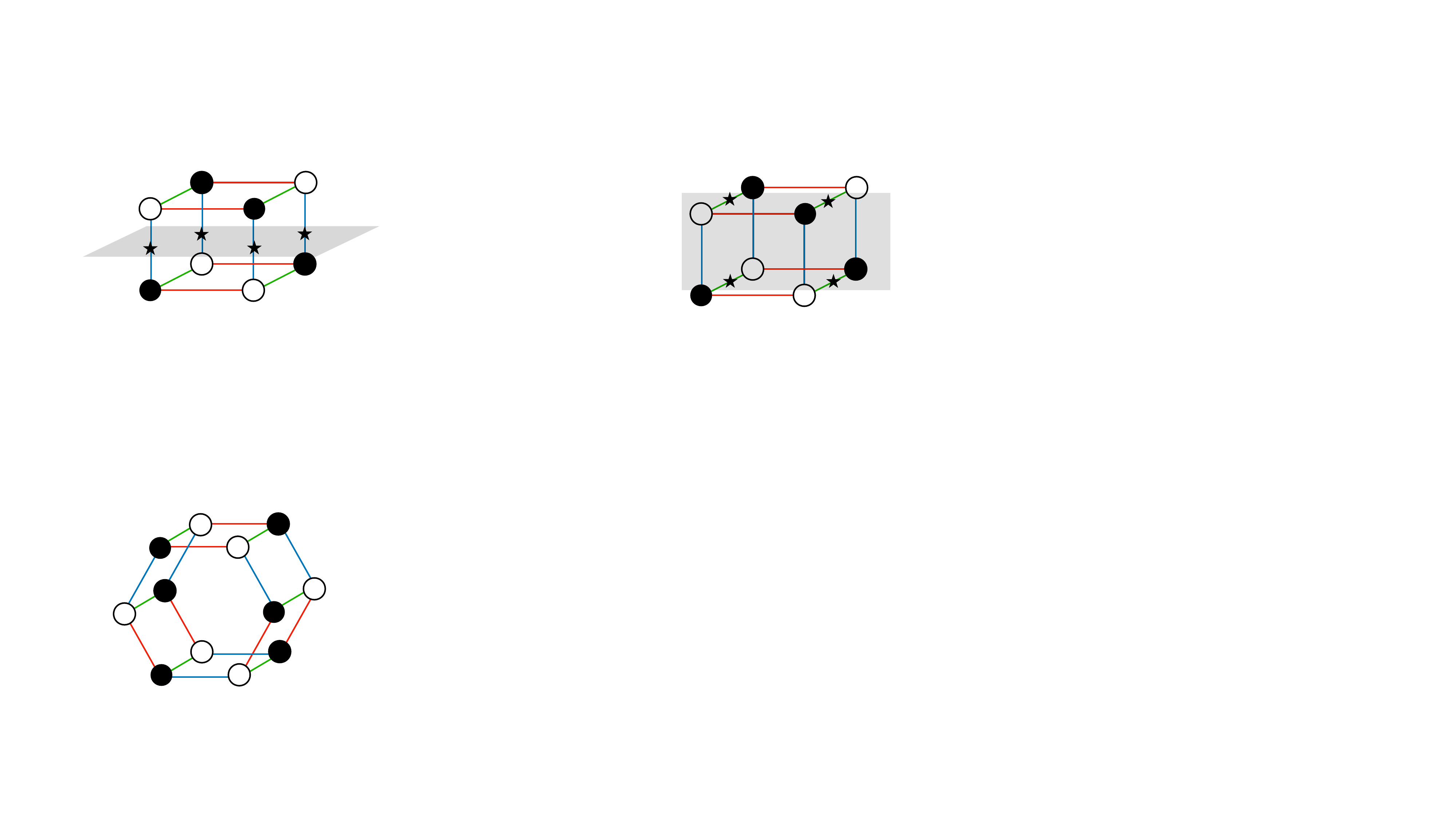}
        \includegraphics[scale=0.25]{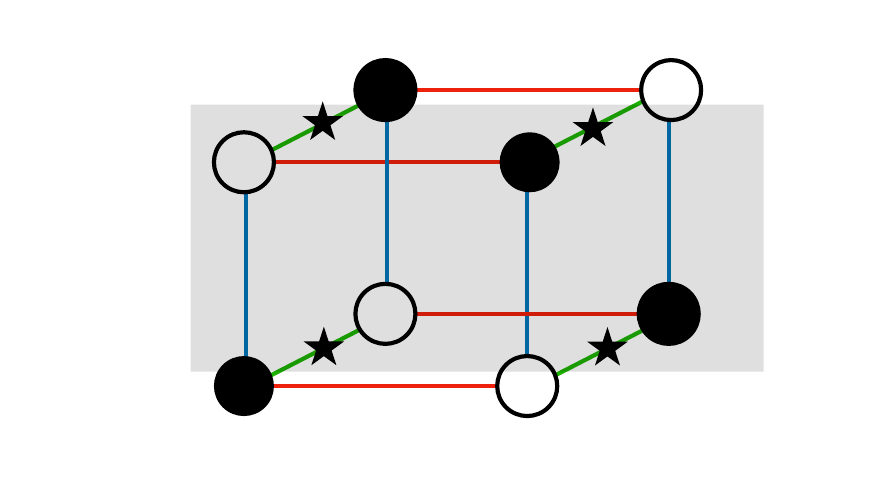}
    \end{center}
    \caption{Graph of ${\cal E}^{(2)}$ and the two reflecting cuts are suggestively denoted by the gray sheets. The first one is $k_1$ and the second one is $k_2$.}\label{cut3}
\end{figure}
\begin{align}
    \CP^{(k_1)}= \begin{pmatrix}
        1 & 1 \\
        1 & 1 
        \end{pmatrix}
        \qquad 
        \CP^{(k_2)}= \begin{pmatrix}
            1 & 0 \\
            0 & 1 
            \end{pmatrix}.
\end{align}
This is not a unique solution for $\CP^{(k_1)}$ and $\CP^{(k_2)}$. We also have a symmetric solution.
\begin{align}
    \CP^{(k_{1})}= \CP^{(k_{2})}=
    \begin{pmatrix}
        1 & \frac12 \\
        \frac12 & 1 
        \end{pmatrix}.
\end{align}

\subsubsection{A new example}\label{new-ex}
Consider the edge-reflecting graph $\CC_n$, shown in figure \ref{C3}. We will  show that it is edge-convex. It is a cyclic graph with $n$ $\psi$s and $n$ $\bar\psi$s with edges alternating between types $A$ and $B$. The edges $A$ and $\bar A$ are shown in the figure with the colors red and blue respectively. Thought of as a local unitary invariant for two parties $A$ and $\bar A$, $\CC_n$ is nothing but ${\rm Tr} \rho_A^n$ or equivalently ${\rm Tr} \rho_{\bar A}^n$.
\begin{figure}[h]
    \begin{center}
        \includegraphics[scale=0.3]{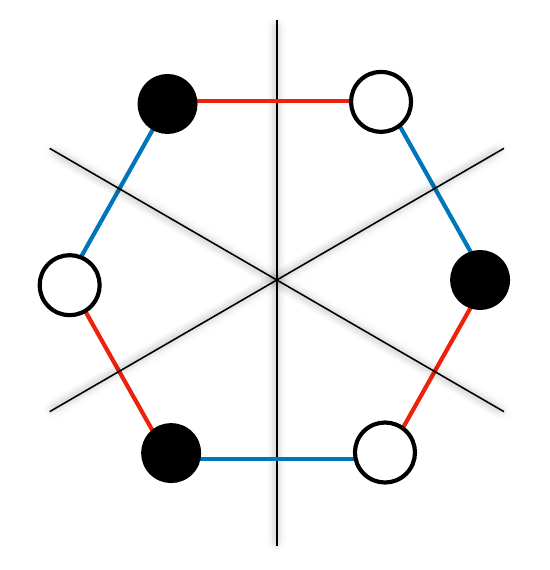}
    \end{center}
    \caption{Graphical presentation of $\CC_3$. It is edge-reflecting. We have denoted the three reflecting cuts with black lines. The graph for $\CC_n$ is cyclic with $n$ $\psi$s and $n$ $\bar\psi$s with edges alternating between types $A$ and $B$. It has $n$ reflecting cuts.}\label{C3}
\end{figure}
We will now show that $\CC_n$ is edge-convex.

We have $\CC_2=\CE^{(\2)}$. We have already discussed the edge-convexity for this. The graph $\CC_n$ has $n$ reflecting cuts, as shown in figure \ref{C3}. For every reflecting cut, we take the positive matrix $\CP^{(k)}$ to be the identity matrix.  Given a pair of $A$-type edges, there is precisely one cut such that these edges are images of each other under it. This ensures that $\CP^{(k)}= {\mathbb I}$ for all $k$ indeed solve the equation \eqref{sol-convex}. 
This proves that $\CC_n$ are edge-convex and shows that $1-{\CC}_n$ is a {PSEM}. In the more familiar notation, this {PSEM} is $1-({\rm Tr}\rho_{\bar A}^n)$.

\subsubsection{Strengthening convexity}\label{strengthen}
In this section we will prove theorem \ref{theorem-psem} that is strictly stronger than proposition \ref{main1}.
\theorempsem*
\begin{proof}
Because an $A$-edge-convex graph is $A$-edge-reflecting,  thanks to lemma \ref{transitive},  it is also $A$-edge-transitive. As a result, an edge-convex graph is  all-edge-transitive.
Due to $A$-edge-transitivity, we have $\delta_e \CZ=\delta_{e'} \CZ$ for any two $A$-edges $e$ and $e'$. Consider the state $|\tilde \CT^{(k)}\rangle :=(\sum_e a^{(k)}_e|\CT_e^{(k)}\rangle) \otimes |\CT^{(k)}\rangle$ and use the inequality
\begin{align}
    \sum_k \langle \tilde \CT^{(k)}|({\mathbb I}-{\mathbb P})|\tilde \CT^{(k)}\rangle \geq 0.
\end{align}
where $P$ is a unitary operator that acts as a swap on the factors involved in the direct product state $|\psi\rangle$. This gives,
\begin{align}
    &\CZ\times \Big(\sum_k \sum_{e\in R_k, e'\in L_k} \CM^{(k)}_{e,e'}\, \delta_e\,\delta_{e'}\CZ\Big)\notag\\
    &\geq \sum_k \sum_{e\in R_k, e'\in L_k} \CM^{(k)}_{e,e'} \, \delta_e\CZ \,\delta_{e'}\CZ = \sum_{e,e'}^{e\neq e'} \delta_e\CZ \,\delta_{e'}\CZ,\notag\\
    &\CZ \delta^2 \CZ \geq |\delta \CZ|^2-\sum_e |\delta_e\CZ|^2 = (1-\frac{1}{n})|\delta \CZ|^2.
\end{align}
For the equality in the second line, we have used the property of the solution \eqref{sol-convex} and for the equality in the last line, we have used the $A$-edge transitivity. This shows that if a graph $\CZ$ is $A$-edge transitive and admits solution to \eqref{sol-convex} for all $A$ then $\hat \nu = 1-\hat \CZ$ is a {PSEM}. 
\end{proof}
\noindent
As remarked earlier $\hat \nu$ is a stronger monotone than $\nu$ because $\nu$ is a composite of $\hat \nu$. Thanks to theorem \ref{cartesian}, we now have stronger {PSEM}s $1-{\hat \CE}^{(2)}, 1-{\hat \CE}^{(3)}$ and $1-{\hat \CC}_{n}$.

The extension of theorem \ref{theorem-psem}  to disconnected edge-convex graphs is as follows. If the $\psi$-graph $\CZ$ is disconnected with two edge-convex connected components $\CZ_1$ and $\CZ_2$ then $\nu = 1-\hat{\CZ_1} \cdot \hat{\CZ_2}$ is a {PSEM}. This follows from a straightforward generalization of the proof to disconnected graphs. 

Note that even though $\CZ=\CZ_1\cdot \CZ_2$, the normalized versions $\hat \CZ \neq \hat{\CZ_1} \cdot \hat{\CZ_2}$. Letting $\nu_i=1-\hat {\CZ_i}$, we see $\nu = \nu_1 +\nu_2-\nu_1 \nu_2$. In particular, $\nu$ is a concave, monotonic function of $\nu_1$ and $\nu_2$ and hence is a composite. As a result, due to theorem \ref{composite}, the bound following from $\nu$ is  weaker than the one following from the set $\{\nu_1, \nu_2\}$. This is the reason we take $\CZ$ to be a connected graph.

\subsubsection{Aside: Conversion of multiple copies}
Let us consider the problem of converting $k$ copies of $|\psi\rangle$ to $k$ copies of $|\phi\rangle$. The bound on the transition probability is provided by
\begin{align}
    \frac{\nu(|\psi\rangle^{\otimes k})}{\nu(|\psi\rangle^{\otimes k})} :=p_k
\end{align}
If $p_1<1$, any meaningful bound on this transition probability must have the property that $p_k\geq p_1^k$  because the conversion of $k$ copies of $|\psi\rangle$ to $k$ copies of $|\phi\rangle$ can certainly done copy
by copy in $k$ steps. In the same vein, if $k$ is a multiple of $l$, then it must be that $p_k\geq p_l^{k/l}$. 
This is indeed the property of $\nu$ and $\hat \nu$ defined in proposition  \ref{main1} and theorem \ref{theorem-psem} respectively. In fact, it obeys stronger inequality, $p_k\geq p_l$ for $k\geq l$. The proof is straightforward and uses remark \ref{factor}.

\section{Edge-convexity of the box product}
In this section we will prove  theorem \ref{cartesian}. Let us first define the box product of two $\psi$-graphs. This definition is very similar to the standard cartesian product of graphs. 

Let us first recall the definition of the cartesian product of ordinary graphs $G_i$'s. If the vertex and edge set of $G_i$ is denoted as $V_i$ and $E_i$ then the vertex set of their cartesian product $G_1\Box G_2$ is the direct product $V_1\otimes V_2$. Two vertices $(v_1,v_2)$ and $(u_1,u_2)$ are joined if either $u_1=v_1$ and $u_2$ and $v_2$ are adjacent to each other in $G_2$ or $u_2=v_2$ and $u_1$ and $v_1$ are adjacent to each other in $G_1$. In other words, the adjacency matrix $A$ of $G_1\Box G_2$ is
\begin{align}\label{cart-adj}
    A= A_1\otimes {\mathbb I} +{\mathbb I}\otimes A_2
\end{align}
where $A_i$ is the adjacency matrix of $G_i$. The notion of cartesian product is extended straightforwardly for $\psi$-graphs. Because we need $\CZ_1{\hat \Box} \CZ_2$ to be a $\psi$-graph, it is necessary to assume that the edges of $\CZ_1$ and $\CZ_2$ do not have a common label. Then the vertex parity of $(v_1,v_2)$ is the sum of parities of $v_1$ and $v_2$. As for the edge-labels, we simply used the formula \eqref{cart-adj} separately for every label $A$ in defining the adjacency matrix of $\CZ_1{\hat \Box} \CZ_2$. See figure \ref{carte} for an example.
\begin{figure}[h]
    \begin{center}
        \includegraphics[scale=0.15]{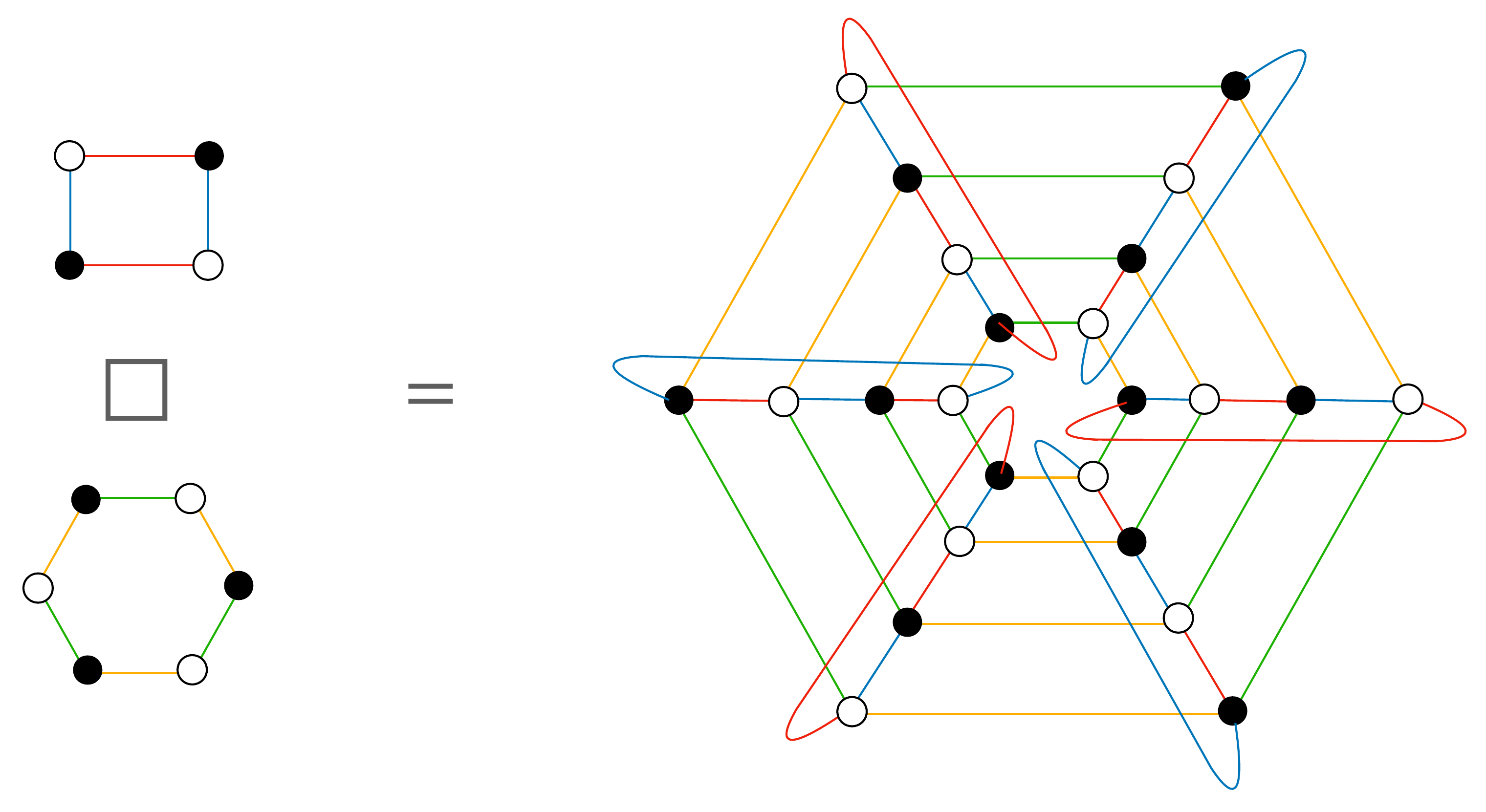}
    \end{center}
    \caption{Example of the box product of $\psi$-graphs.}\label{carte}
\end{figure}

Before we prove theorem \ref{cartesian}, we need to define the notion of vertex-convexity.
\begin{definition}
    A $\psi$-graph is called vertex-convex if it admits solution to condition:
    \begin{align}\label{vertex-sol}
        \sum_{k\, {\rm s.t.}\,v\in R_k, v'\in L_k} \,\,\CM_{v,v'}^{(k)}=1. \qquad \forall\,  v,v'.
    \end{align}
    where the sum is over all reflecting cuts that separate vertices $v$ and $v'$. The matrix $\CP^{(k)}(v,v')= \CM^{(k)}_{v, k(v')}$ is positive semi-definite for all $k$'s in the sum.
\end{definition}
\noindent
This parallels the definition of edge-convexity but for vertices.

We have the following lemma that relates the edge and vertex notions of convexity.
\begin{restatable}{lem}{vertexconvex}\label{vertex-convex}
    If a $\psi$-graph is edge-convex then it is vertex convex.
\end{restatable}
\noindent
The proof is given in appendix \ref{proof-theorem1}.

\cartesian*
\begin{proof}
    First let us show that if $\CZ_1$ and $\CZ_2$  are edge-reflecting then $\CZ_1{\hat \Box} \CZ_2$ is also edge-reflecting. Without loss of generality let us assume that $\CZ_1$ is $A$-edge-reflecting and that $\CZ_2$ does not have any $A$-edge. The $A$-edges in the cartesian product will be labeled by the pair $(e,u)$ and $(f,v)$ where $e$ and $f$ are $A$-edges of $\CZ_1$ and $u, v$ are vertices of $\CZ_2$. If $e$ and $f$ are distinct edges of $\CZ_1$, we can simply use the reflecting cut that separates them in $\CZ_1$ and take its cartesian product with $\CZ_2$ to produce the required reflecting cut in $\CZ_1{\hat \Box} \CZ_2$. If $e=f$ and $u\neq v$, we find a reflecting cut in $\CZ_2$ separating $u$ and $v$ using its vertex-reflecting property. Taking its cartesian product with $\CZ_1$ produces the required reflecting cut. 

    Now we construct the solution to \eqref{sol-convex} for $\CZ_1{\hat \Box} \CZ_2$ using a similar strategy. 
    using the solution to the same equation for $\CZ_1$ and $\CZ_2$ separately. We would like to show that the pair of $A$-edges, $(e,u)$ and $(f,v)$ appears with weight $1$. Denoting the reflecting cuts separating $e$ and $f$ in $\CZ_1$ as $k$, 
    \begin{align}
        \CM^{(\tilde k)}_{(e,u),(f,v)}= \CM^{(k)}_{e,f}\qquad \forall u,v.
    \end{align}
    Here  $\tilde k = k{\hat \Box}\, \CZ_2$ is a reflecting cut of 
    $\CZ_1{\hat \Box} \CZ_2$. 
    As $\CP_{e,f}$ is positive definite, $\CP_{(e,u),(f,v)}$ is also positive definite. This will create all the pairs of $A$-edges with weight $1$ except for the ones where $e=f$ and $u\neq v$. To produce these remaining terms with weight $1$, we use the vertex-convexity of graph $\CZ_2$. Let us denote the reflecting cuts separating $u$ and $v$ of $\CZ_2$ as $k_V$. 
    \begin{align}
        \CM^{(\tilde k)}_{(e,u),(f,v)}= \CM^{(k_V)}_{u,v}\, \delta_{e,f}.
    \end{align}
    Here  $\tilde k = \CZ_1\Box k_V$ is a reflecting cut of 
    $\CZ_1{\hat \Box} \CZ_2$. 
    For an edge-convex graph, the matrix $\CM^{(k_V)}_{u,v}$ is constructed in lemma \ref{vertex-convex}.
\end{proof}

The examples of edge-convex graphs that we previously constructed with brute  force viz. ${\cal E}^{(2)}$ and ${\cal E}^{(3)}$ can be understood as simple consequence of theorem \ref{cartesian}. Note that ${\cal E}^{(\tq)}={\cal E}_{1}^{{\hat \Box } q}=:{\cal E}_{\tq;\emptyset}$. This shows that not only for $\tq=2,3$ but for all $\tq\in {\mathbb Z}_+$, ${\cal E}^{(\tq)}$ is edge-convex. Combined with the result that ${\cal C}_n$ is edge convex, we get the most general edge-convex graph
\begin{align}
    {\cal E}_{1}^{{\hat \Box } m} {\hat \Box} \CC_{n_1}{\hat \Box}\ldots {\hat \Box} \CC_{n_k}=:{\cal E}_{m;n_1, \ldots, n_k}. 
\end{align}
In the next section we will explore the transformation probability bounds following from ${\cal E}_{1;n}$.

\section{Example of a bound}
In this section we will compute the bounds on  transition success probability for tri-partite states $|\psi\rangle \to |\phi\rangle$. We will use the {PSEM} $\hat \nu ({\tilde C}_n)$ where  ${\tilde C}_n:= {\cal E}_{1;n}$. The  $\psi$-graph of ${\tilde C}_n$ is shown in figure \ref{sol-example}. 
\begin{figure}[t]
    \centering
        \includegraphics[width=4cm]{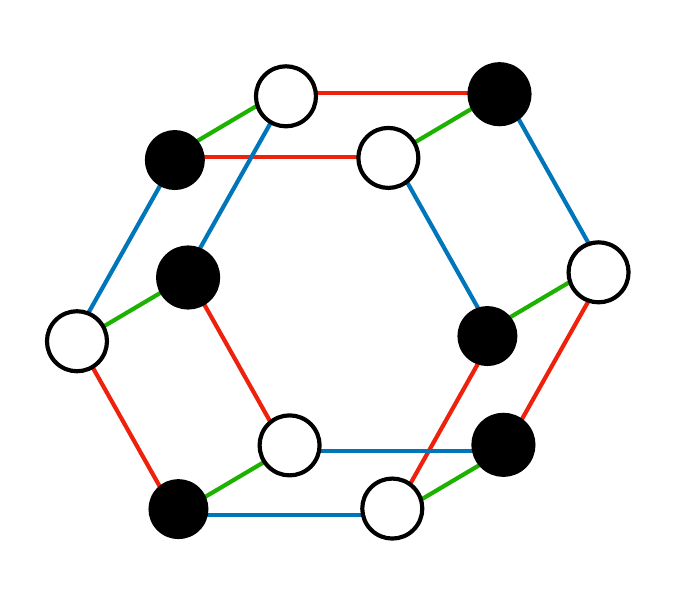}
    \label{fig:my_label}
    \caption{The $\psi$-graph of multi-invariant ${\tilde C}_3:= {\CZ}_{I_3\sqcup A_1}$.}\label{sol-example}
\end{figure}
In particular, we will take $|\psi\rangle$ to be a GHZ-type state
\begin{equation}
    |\psi\rangle = \frac{1}{\sqrt {5}}(2 |000\rangle + |111\rangle)
\end{equation}
Let us denote the three qubits as $A,B,C$. As discussed near equation \eqref{slocc}, in order to have a non-zero transition probability to $|\phi\rangle$, $|\phi\rangle$ needs to be related to $|\psi\rangle$ by an \slocc\ transformation. We take $|\phi\rangle$ to be,
\begin{align}
    |\phi\rangle = M_1\otimes M_2\otimes M_3 |\psi\rangle|_{\rm normalized}.
\end{align}
For convenience, we will take $M_1, M_2$ and $M_3$ to be identical and equal to $\exp(\alpha K)$ for some real $\alpha$. We choose $K$ arbitrarily to be,
\begin{equation}
    K = \begin{pmatrix}
        1 & 1 \\ -2 & -1
    \end{pmatrix}.
\end{equation}
This gives us a family of states $|\phi_\alpha\rangle$, all of which are related to $|\psi\rangle$ by \locc. 
Because all $M_i$ are chosen to be identical, just like $|\psi\rangle$, $|\phi\rangle$ is also symmetric under permutation of parties.
We first compute the optimal transition probability $p_V$ for $|\psi\rangle \to |\phi_\alpha\rangle$, by thinking of them as bi-partite states under the bi-partition, say $A$ and $BC$. Due to the symmetry of the state,  all the bi-partitions are equivalent.
\begin{figure}[t]\label{ghz-bound}
    \centering
        \includegraphics[width=7cm]{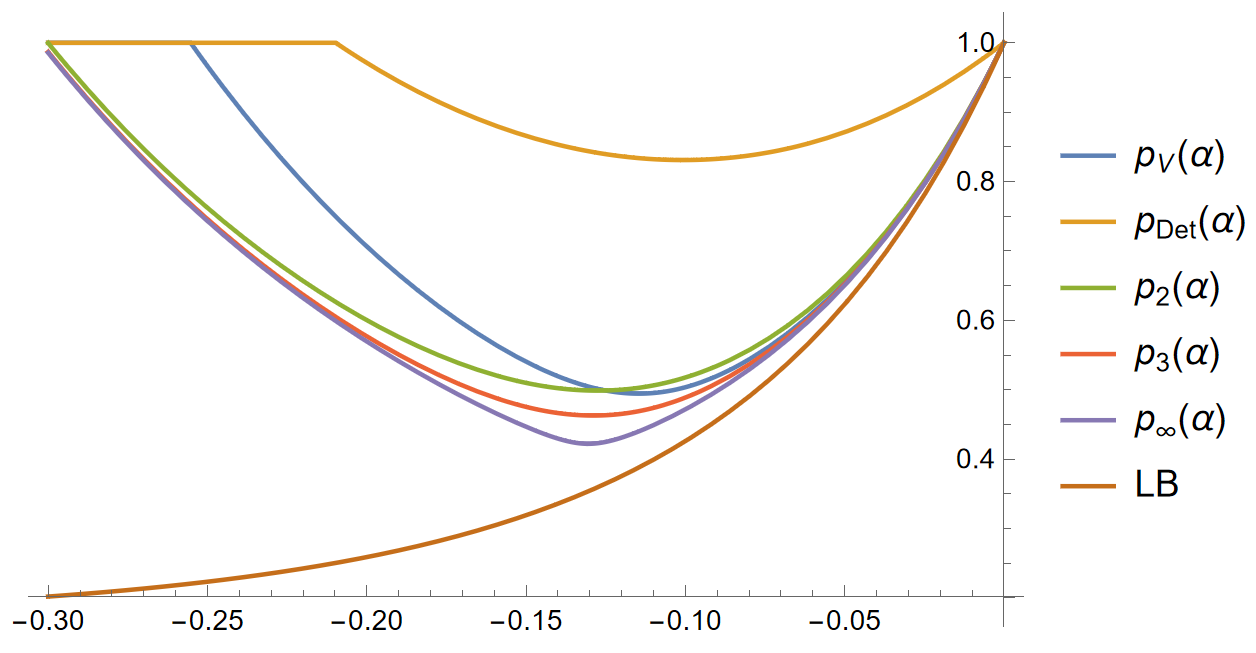}
    \label{fig:my_label}
    \caption{Plot of $p_V(\alpha)$ and $p_n(\alpha)$ for $n = 2,3,4$. We have also plotted the lower bound coming from equation \eqref{lowerb}.}
\end{figure}
\begin{equation}
    p_V(\alpha) = \text{min}_{k} \frac{\Tilde{\nu}_k^V(|\psi\rangle)}{\Tilde{\nu}_k^V(|\phi_\alpha\rangle)} .
\end{equation}
This is the optimal transition probability of $|\psi\rangle \to |\phi_\alpha\rangle$ if qubits $B$ and $C$ can be brought together and a joint operation can be performed on them. 
In the absence of such a joint operation, we expect that this transition probability to be smaller. We use the tri-partite monotones $\nu_n :=\hat \nu ({\tilde C}_n)$, to put upper bound $p_n$ on this ``tri-partite'' probability. 
\begin{align}
    p_n(\alpha) = \frac{\nu_n(|\psi\rangle)}{\nu_n(|\phi_\alpha\rangle)}.
\end{align}
We have plotted $p_V(\alpha)$ as well as $p_n(\alpha), n= 2, 3, \infty$ for a  range of $\alpha$ in figure \ref{ghz-bound}. 
This figure effectively illustrates the central point of the paper: we see that for a range of $\alpha$'s, the transition probability for $|\psi\rangle \to |\phi_\alpha\rangle$ that is achievable by truly local processes is smaller than the probability if the processes were non-local.

We have also plotted the bound on the transition probability obtained from the {PSEM} $\nu_{\rm Det}$ of equation \eqref{det-psem},
\begin{align}
    p_{\rm Det}(\alpha) = \frac{\nu_{\rm Det}(|\psi\rangle)}{\nu_{\rm Det}(|\phi_\alpha\rangle)}.
\end{align}
We see that the bounds obtained from the {PSEM}s constructed here are more stringent than $p_{\rm Det}(\alpha)$ bound for the given range of $\alpha$.
As a sanity check we have also plotted the lower bound on the transition probability $|\psi\rangle\to |\phi_\alpha\rangle$ computed in equation \eqref{lowerb}. As expected, it is below all the upper bounds discussed here. 

It would be interesting to investigate if the multi-partite entanglement monotones constructed here are primitive as defined in Definition \ref{primitive}. Another important question, along these lines, is to find the minimal complete set of  multi-partite entanglement monotones and show that it is minimal by explicitly constructing the optimal transformation protocol. 

\section{Conclusions and outlook}
\label{sec:conclusions}

We introduced a constructive route to multipartite pure-state entanglement monotones starting from local-unitary invariant polynomials namely multi-invariants and their bipartite, edge-labeled graph representation.
Our main technical result is a simple \emph{sufficient} condition - edge-convexity - on the graph that ensures that the appropriately normalized  multi-invariants are convex with respect to every single-party partial trace, and hence yield entanglement monotones via Vidal's concavity  characterization as stated  in theorem \ref{concave}.
We also established closure of the construction under a natural box-product operation, providing scalable families of explicit multipartite monotones.

Operationally, each PSEM immediately supplies an upper bound on the optimal success probability of stochastic LOCC conversion between pure states via the ratio of its values on the initial and the target states \cite{Vidal_1999, Vidal_2000}.
It is known that a convex monotonic function of PSEMs is again a PSEM \cite{Szalay_2015}. In theorem \ref{composite}, we showed that the bound on the success probability obtained from such composite monotones cannot improve upon the bounds obtained from the primitive monotones. 

From a practical standpoint, the multi-invariants underlying PSEMs admit multi-copy representations of the form
$Z_{\vec\sigma}=\Tr[\rho^{\otimes k}U_{\vec\sigma}]$ with $U_{\vec\sigma}$ a tensor product of party-wise permutations across $k$ copies; consequently they are, in principle, accessible using ancilla-assisted interferometric  estimation of $\Tr(\rho^{\otimes k}U_{\vec\sigma})$~\cite{Ekert:2002qtj, LeiferLindenWinter2004MeasuringPolynomialInvariants}
or randomized-measurement primitives based on unitary designs~\cite{Elben2018RenyiRandomQuenches,Brydges2019ProbingRenyiRandomized,Huang2020ClassicalShadows}. Below we list some future research directions.

\begin{itemize}
    \item Although our focus is on pure states, the PSEMs are constructed from simple polynomial functions of the state. We expect that this simplicity would make them amenable to extension to mixed states via the convex roof construction. 
    Thinking of the convex roof as a minimization problem over unitaries acting on the purifier, the convex roof can be upper-bounded  by the Haar average over the unitary group. For polynomial functions, such averaging can be performed using Weingarten calculus and the result is expressed as a finite sum over permutations. We leave a systematic development of such convex-roof bounds for future work.
    \item We proved that the $\psi$-graphs ${\cal E}_1$ and ${\cal C}_n$ are edge-convex, and that edge-convexity is preserved under the box product. This yields infinite families of edge-convex multi-invariants and hence PSEMs. A natural next step is to classify edge-convex $\psi$-graphs, thereby enlarging and organizing the space of monotones produced by our construction.
    \item Our multi-invariants admit a natural definition in quantum field theory via replica constructions. It would be interesting to clarify their physical interpretation in that setting. In particular, since renormalization-group (RG) flow provides a notion of coarse-graining reminiscent of locality-restricted processing, one may ask whether QFT analogues of our PSEMs can be combined into quantities that are monotonic along RG trajectories. We view this as an intriguing direction for future investigation.
\end{itemize}

\section*{Acknowledgements}
We would like to thank Jonathan Harper, Vineeth Krishna, Gautam Mandal, Shiraz Minwalla, Onkar Parrikar, Trakshu Sharma, Piyush Shrivastava, Sandip Trivedi for interesting discussions. 
This work is supported by the Infosys Endowment for the study of the Quantum Structure of Spacetime and by the SERB Ramanujan fellowship.  We acknowledge the support of the Department of Atomic Energy, Government of India, under Project Identification No. RTI 4002. HK would like to thank KVPY DST fellowship for partially supporting his work.
Finally, we acknowledge our debt to the people of India for their steady support to the study of the basic sciences.

\appendix

\section{Pure states to mixed states}\label{convex-roof}
In this appendix we discuss the extension of pure state entanglement monotones to mixed states. 
A pure state entanglement monotone is extended to mixed states using the so called \emph{convex roof extension} \cite{Uhlmann1998}. We first write the given mixed state as an ensemble of pure states $\rho=\sum_i p_i |\psi_i\rangle \langle \psi_i|$. 
There are multiple pure-state ensembles that realize $\rho$. Let us denote their set as ${\cal E}_\rho$. The convex roof function ${\tt CR}_\nu(\rho)$ on mixed states is defined as,
\begin{definition}[Convex roof]
    \begin{align}\label{roof}
        {\tt CR}_\nu(\rho):= \min_{{\cal E}_\rho} \sum_i p_i\, \nu(|\psi_i\rangle), \qquad \rho=\sum_i p_i |\psi_i\rangle \langle \psi_i|.
    \end{align}
\end{definition}
\noindent 
For future use, let us denote the pure state ensemble realizing the minimum as ${\tt e}_\nu(\rho)$. This extension was first considered by Uhlmann in  \cite{Uhlmann1998}.
The convex roof extension of the entanglement entropy was one of the first mixed state entanglement measures that was constructed and studied \cite{Bennett_1996}. This measure was termed the entanglement of formation. 
\begin{lemma}[Horodecki \cite{10.5555/2011326.2011328}]\label{horodeckia}
    ${\tt CR}_\nu(\rho)$ is a mixed state entanglement monotone.
\end{lemma}
\begin{proof}
    First note that by virtue of being a convex roof ${\tt CR}_\nu(\rho)$ is convex. This is because the union of sets ${\cal E}_{\rho_1}$ and ${\cal E}_{\rho_2}$ over which we do minimization to compute $p {\tt CR}_\nu(\rho_1)+ (1-p) {\tt CR}_\nu(\rho_2) $ is a subset of ${\cal E}_{p\rho_1+(1-p)\rho_2}$. 
    
    Now we will show that ${\tt CR}_\nu(\rho)$ obeys the equation \eqref{locc-mono}.
    Let $\rho$ be an ensemble of pure states as $\rho=\sum_i p_i |\psi_i\rangle\langle \psi_i|$. Its transformation with respect to the set $E_k$ produces
    \begin{align}
        \Lambda(\rho) &= \sum_k q_k \sigma_k,\\
         q_k &:={\rm Tr}(E_k \rho E_k^\dagger),\qquad \sigma_k := E_k \rho E_k^\dagger/q_k. \notag
    \end{align}
    Using the expression for $\rho$ in terms of pure states,
    \begin{align}\label{sigmaconvex}
        \sigma_k &=\frac{1}{q_k} \sum_i p_i q_k^i|\psi_i^k\rangle \langle \psi_i^k|, \\
        q_k^i &:= {\rm Tr}(E_k |\psi_i\rangle \langle\psi_i| E_k^\dagger), \qquad |\psi_i^k\rangle := E_k |\psi_i\rangle/\sqrt{q_i^k}.\notag 
    \end{align}
    We want to show that ${\tt CR}_\nu(\rho) \geq \sum_k q_k {\tt CR}_\nu(\sigma_k)$. As $\sigma_k$ is given as the convex combination as in equation \eqref{sigmaconvex}, using convexity of ${\tt CR}_\nu$,
    \begin{align}\label{temp1}
        \sum_k q_k {\tt CR}_\nu(\sigma_k)\leq \sum_k \sum_i p_i q_k^i \nu(|\psi_i^k\rangle)\leq \sum_i p_i \nu(|\psi_i\rangle).
    \end{align} 
    In the last inequality, we have used the pure state monotonicity \eqref{locc-mono-pure}. As this equation is true for any pure state ensemble realizing $\rho$, it is also true for the optimal pure state ensemble  ${\tt e}_\nu (\rho)$. Taking the optimal ensemble, the right-hand side of the inequality \eqref{temp1} is precisely ${\tt CR}_\nu (\rho)$.
\end{proof}
Let us also comment that there are mixed state entanglement monotones that are not defined using convex roofs. 
They include an appropriately defined distance of $\rho$ from the set of separable states e.g. \cite{Vedral_1997, Vedral_1998}. Just like convex roofs, these measures are also difficult to compute.  In the bi-partite case, there is a monotone that is efficiently computable - even for mixed states - known as logarithmic negativity \cite{Plenio_2005}. For two qubit mixed states, the convex roof of the entanglement entropy, called entropy of formation, was computed explicitly in \cite{Wootters:1997id}.

\section{Review of Proofs}\label{review-proofs}

\szalay*
\begin{proof}
    Let the pure state $|\psi\rangle$ transform into an ensemble of pure states $ \Lambda(|\psi\rangle)=\{p_i, |\psi_i\rangle\}$ under the \locc\ transformation $\Lambda$. 
    \begin{align}
        &\sum_i p_i G(\nu_1(|\psi_i\rangle),\ldots, \nu_k(|\psi_i\rangle))\notag\\ 
        &\leq G(\sum_i p_i \nu_1(|\psi_i\rangle), \ldots, \sum_i p_i \nu_k(|\psi_i\rangle))\notag\\
        &\leq G(\nu_1(|\psi\rangle), \ldots, \nu_k(|\psi\rangle)).
    \end{align}
    In the first line we have used the concavity of $G$ in all of its arguments. In the second line, we have used the monotonicity of $G$, along with the defining {PSEM} property \eqref{locc-mono-pure} of $\nu_i$s.
\end{proof}

\vidala*
\begin{proof}
    Let $\rho_*$ be one of the possible outcomes of the positive operator valued measurement on $\rho$. In other words, $\rho_*$ appears as proportional to one of the $E_i\rho E_i^\dagger$, say the one corresponding to $E_*$ and the probability of this outcome is $p_*={\rm Tr} (E_*\rho E_*^\dagger)$. Then, using the properties of entanglement monotones,
    \begin{align}
        \mu(\rho) \geq  p_*\mu(\rho_*)+ \sum_{i\in{\rm rest}} p_i  \mu(\rho_i) \geq p_*\mu(\rho_*).
    \end{align}
    This shows that $p_{\rho\to \rho_*}\leq \frac{\mu(\rho)}{\mu(\rho_*)}$ for every $\mu$. As for the equality in equation \eqref{opt}, note that $p_{\rho\to \rho_*}$ itself is an entanglement monotone of $\rho$, for a fixed $\rho_*$, and that $p_{\rho_*\to \rho_*}=1$.
\end{proof}

\vidalb*
\begin{proof}
    Let the state $\rho=|\psi\rangle\langle \psi|$ be converted to 
        \begin{align}\label{partyA}
            &\Lambda(\rho)=\sum_i p_i |\psi_i\rangle \langle \psi_i|,\notag\\
            &  {p_i}:=|E_i |\psi\rangle|^2, \,\psi_i := E_i |\psi\rangle /\sqrt{p_i}
        \end{align}
    after a local operation by party $A$.  We need to show that 
    \begin{align}
        f(|\psi\rangle) \geq \sum_i p_i f(|\psi_i\rangle).
    \end{align}
    The operation \eqref{partyA} keeps $\rho_{\bar{A}} $ invariant i.e. $\Lambda(\rho)_A=\rho_{\bar{A}}$ because of the trace preserving property $\sum_i E^{(A)\dagger}_i E^{(A)}_i ={\mathbb I}$. So we have
    \begin{align}
        \rho_{\bar{A}} = \sum_i p_i \rho_{i,{\bar{A}}},\qquad \rho_{i,{\bar{A}}}:= {\rm Tr}_A |\psi_i\rangle\langle \psi_i|.
    \end{align}
    The concavity of $f$ in $\rho_{\bar{A}}$ implies, $f(\rho_{\bar{A}}) \geq \sum_i p_i f(\rho_{i,{\bar{A}}})$. When expressed in terms of purifications $|\psi\rangle$ and $|\psi_i\rangle$'s, this equation is exactly the one that we want to prove. Note that, for this argument to go through, it is necessary that the density matrices $\rho_{\bar{A}}$ and $\rho_{i,{\bar{A}}}$ are obtained by tracing out precisely one party and not more.
    As the concavity holds for partial trace $\rho_{\bar{A}}$ for any party $A$, the monotonicity also holds for local operations by any party. 
\end{proof}

\vidalc*

\begin{proof}
    We will show that $\tilde \nu_k^V$ are concave function of the density matrix $\rho_{\bar{A}}$. This, in conjunction with theorem \ref{concave} proves that $\tilde \nu_k^V$ are {PSEM}s.
It is convenient to use the following characterization of $\sum_{i=1}^k \lambda_i$. 
\begin{align}
    \sum_{i=1}^k \lambda_i = \max_{P_k} \,{\rm Tr} (P_k \, \rho_{\bar{A}}),
\end{align}
where $P_k$ are rank $k$ projectors and maximization is performed over all such projectors. This fact is known as Ky Fan's maximum principle. Clearly,
\begin{align}
    &\max_{P_k} \,{\rm Tr} P_k \,(  p\rho_{1,{\bar{A}}}+(1-p)\rho_{2,{\bar{A}}}) \notag\\
    &\leq  p\max_{P_k} \,{\rm Tr} (P_k \, \rho_{1,{\bar{A}}})+ (1-p)\max_{P_k} \,{\rm Tr} (P_k \,\rho_{2,{\bar{A}}}).
\end{align}
This is because, the right-hand side allows for two different choices of projectors to maximize the two terms but on the left-hand side, the same projector has to be used for both terms. This shows that $\sum_{i=1}^k \lambda_i$ is a convex function of $\rho$ and hence $\tilde \nu_k^V$ is a concave function of $\rho$. The constant $1$ is chosen so that it vanishes on the factorized state i.e. for rank $1$ density matrix state with $\lambda_1=1$ and the rest of the $\lambda$'s zero.
\end{proof}

\section{Proofs of new results}\label{proof-theorem1}

\composite*
\begin{proof}
    We will prove this inequality using induction in $k$. Let first us assume that we have proved it for $k=1$. Consider the case $x_k<x'_k$, then
    \begin{align}
        &\frac{G(x_1,\ldots, x_k)}{G(x'_1,\ldots, x'_k)}=\frac{G(\alpha_1\, x_k,\ldots, \alpha_{k-1}\, x_k, x_k)}{G(\alpha_1 x'_k,\ldots, \alpha_{k-1} \, x'_k, x'_k)}\notag\\
        &\times \frac{G(\alpha_1 x'_k,\ldots, \alpha_{k-1} \, x'_k, x'_k)}{G(\beta_1 \, x'_k,\ldots, \beta_{k-1}\, x'_k, x'_k)}
    \end{align} 
    Here we have defined $\alpha_i= x_i/x_k$ and $\beta_i=x'_i/x'_k$ for $i=1,\ldots, k-1$.
    Let the first ratio on the right hand side be $A$ and the second be $B$. Define $f(x_k)=G(\alpha_1\, x_k,\ldots, \alpha_{k-1}\, x_k, x_k)$. This is a monotonic and concave function of $x_k$. This shows 
    \begin{align}\label{f1}
        A \geq \min\Big( \frac{x_k}{x'_k},1\Big)=\frac{x_k}{x'_k}.
    \end{align}
    Now define $\tilde f(\alpha_1,\ldots, \alpha_{k-1})=G(\alpha_1 x'_k,\ldots, \alpha_{k-1} \, x'_k, x'_k)$. This is a monotonic and concave function of its $k-1$ arguments. Using induction, 
    \begin{align}\label{f2}
        B \geq \min\Big(
            \frac{\alpha_1}{\beta_1},\ldots, \frac{\alpha_{k-1}}{\beta_{k-1}},1
            \Big).
    \end{align}
    The inequalities \eqref{f1} and \eqref{f2} imply
    \begin{align}
        A\cdot B \geq \min \Big(
            \frac{x_1}{x'_1},\ldots, \frac{x_k}{x'_k},1
            \Big).
    \end{align}
    In the other case $x_k \geq x'_k$,
    \begin{align}
        &\frac{G(x_1,\ldots, x_k)}{G(x'_1,\ldots, x'_k)} \geq \frac{G(x_1,\ldots,x_{k-1}, x'_k)}{G(x'_1,\ldots,x'_{k-1}, x'_k)} \notag\\
        &\geq \min\Big(
            \frac{x_1}{x'_1},\ldots, \frac{x_{k-1}}{x'_{k-1}},1
            \Big)\notag\\
            &\geq \min\Big(
                \frac{x_1}{x'_1},\ldots, \frac{x_{k}}{x'_{k}},1
                \Big).
    \end{align}\
    Here, in the first inequality, we have used the monotonicity property of $G$ i.e. $G(x_1,\ldots,x_{k-1}, x'_k)\geq G(x_1,\ldots,x_{k-1}, x_k)$. In the second inequality, we have used the inductive assumption for the function $f'(x_1,\ldots, x_{k-1})=G(x_1,\ldots,x_{k-1}, x'_k)$ of 
    $k-1$ arguments.

    Now, let us prove the inequality \eqref{comp} for $k=1$. For $x_1\geq x'_1$,
    \begin{align}
        \frac{G(x_1)}{G(x'_1)}\geq 1,
    \end{align}
    due to monotonicity of $G$. We only need to consider the case $x_1<x'_1$. 
    \begin{align}
        \frac{G(x_1)}{G(x'_1)}-\frac{x_1}{x'_1}=\frac{x'_1 G(x_1)-x_1 G(x'_1)}{x'_1 G(x'_1)}.
    \end{align}
    Now we will show that the numerator $n(x,y)=y G(x)-x G(y)$ of the right hand side is non-negative for $x<y$. 
    \begin{align}
        \partial_{y} n(x,y)&=G(x)-xG'(y)\notag\\
        \partial_x\partial_y n(x,y)&=G'(x)-G'(y)=-\int_x^y d\alpha G''(\alpha)\geq 0.
    \end{align}
    This shows that the function $\partial_y n(x,y)$ is a monotonic function of $x$. So
    \begin{align}
        G(x)-xG'(y)\geq G(0)-0G'(y)\geq 0.
    \end{align}
    This shows that the function $n(x,y)$ is a monotonic function of $y$. Hence $n(x,y)\geq n(x,x)=0$. This completes the proof.
\end{proof}

\generalproj*

\begin{proof}
    Recall that $f(\rho_{\bar{A}})$ can also be written as a function on the pure states $f(|\psi\rangle)$ where $|\psi\rangle$ is a purification of $\rho_{\bar{A}}$. Convexity of $f(\rho_{\bar{A}})$ means that it obeys 
    \begin{align}\label{assume1}
        &f(p \rho^{(1)}_{\bar{A}}+(1-p)\rho^{(2)}_{\bar{A}}) \notag\\
        &\leq p f(\rho^{(1)}_{\bar{A}})+(1-p)f(\rho^{(2)}_{\bar{A}}),
    \end{align}
    where $\rho_{\bar{A}}^{(i)}$ are normalized density matrix. Thanks to linear homogeneity of $f(\rho_{\bar{A}})$, the inequality continues to hold even if we take $\rho^{(i)}_{\bar{A}}$'s to be un-normalized. Lets prove it. 
    Taking $\rho^{(i)}_{\bar{A}}$ to be un-normalized, define its normalized version as $\rho^{(i),N}_{\bar{A}} :=\rho^{(i)}_{\bar{A}}/ n_i  $ where $n_i = {\rm Tr} \rho^{(i)}_{\bar{A}}$. 
    \begin{align}
        &f(p \rho^{(1)}_{\bar{A}}+(1-p)\rho^{(2)}_{\bar{A}})\notag\\
        &=f(n(p \frac{n_1}{n}\rho^{(1),N}_{\bar{A}}+(1-p) \frac{n_2}{n}\rho^{(2),N}_{\bar{A}}))\notag\\
        &= n f(p \frac{n_1}{n}\rho^{(1),N}_{\bar{A}}+(1-p) \frac{n_2}{n}\rho^{(2),N}_{\bar{A}})
    \end{align}
    Here $n$ is ${\rm Tr} (p \rho^{(1)}_{\bar{A}}+(1-p)\rho^{(2)}_{\bar{A}})$. In the second line we have used linearity of $f$. Now we have $f$ evaluated on the normalized density matrix that is expressed as convex linear combination of two normalized density matrices. 
    \begin{align}
        &f(p \rho^{(1)}_{\bar{A}}+(1-p)\rho^{(2)}_{\bar{A}}) \notag\\
        &\leq n \Big(p \frac{n_1}{n} f(\rho^{(1),N}_{\bar{A}})+(1-p) \frac{n_2}{n} f(\rho^{(2),N}_{\bar{A}}) \Big)\notag\\
        & = p f(\rho^{(1)}_{\bar{A}})+(1-p)f(\rho^{(2)}_{\bar{A}}).
    \end{align}
    In the first line we used the convexity of $f$ for normalized density matrices and in the second line we again used its linearity.

    Now we move on to show 
    \begin{align}\label{assume}
        &\tilde f_{k_1,\ldots, k_\tq}(p \rho^{(1)}_{\bar{A}}+(1-p)\rho^{(2)}_{\bar{A}}) \notag\\
        &\leq p \tilde f_{k_1,\ldots, k_\tq}(\rho^{(1)}_{\bar{A}})+(1-p)\tilde f_{k_1,\ldots, k_\tq}(\rho^{(2)}_{\bar{A}}),
    \end{align}
    where $\rho^{(i)}_{\bar{A}}$ are normalized density matrices. Let $|\psi^{(i)}_A\rangle$ be their purifications. Recall,
    \begin{align}
        &\tilde f_{k_1,\ldots, k_\tq}(p \rho^{(1)}_{\bar{A}}+(1-p)\rho^{(2)}_{\bar{A}}) \\
        &:=\max_{P_{k_1,\ldots , P_{k_\tq}}} f(P_{k_1,\ldots , P_{k_\tq}} (\sqrt{p} |\psi^{(1)}\rangle +\sqrt{1-p}|\psi^{(2)}\rangle)).\notag
    \end{align}
    Let  $P^*_{k_1,\ldots, k_\tq}$ be the projector that maximizes the right hand side. Let $P^*_{k_1,\ldots, k_\tq}|\psi^{(i)}\rangle = |\tilde {\psi}^{(i)}\rangle$ be two un-normalized states. Then
    \begin{align}
        &\tilde f_{k_1,\ldots, k_\tq}(p \rho^{(1)}_{\bar{A}}+(1-p)\rho^{(2)}_{\bar{A}}) \notag\\
        &= f(\sqrt{p} |\tilde {\psi}^{(1)}\rangle+ \sqrt{1-p} |\tilde {\psi}^{(2)}\rangle).
    \end{align}
    Using the convexity of $f$ for un-normalized states and linearity,
    \begin{align}
        &\tilde f_{k_1,\ldots, k_\tq}(p \rho^{(1)}_{\bar{A}}+(1-p)\rho^{(2)}_{\bar{A}})\notag\\
        &\leq p f( |\tilde {\psi}^{(1)}\rangle)+ (1-p)  f(|\tilde {\psi}^{(2)}\rangle).
    \end{align}
    Clearly, $f( |\tilde {\psi}^{(i)}\rangle) \leq \max_{P_{k_1,\ldots ,k_\tq}} f(P_{k_1,\ldots ,k_\tq} |\psi^{(i)}\rangle)$. This shows that the right hand side of the above inequality is $\leq p\tilde f(|\psi^{(1)}\rangle)+ (1-p) \tilde f(|\psi^{(2)}\rangle)$.
\end{proof}

\vertexconvex*
\begin{proof}
    We will only make use of $A$-edge-convexity for two edge-labeles $A$ to prove this result. 
    If the graph is $\CE^{(1)}$, it is vertex-convex. Let us assume that it is not $\CE^{(1)}$.
    Let us consider all the reflecting cuts $k$ separating $A$-edges. As proved above, those are also the reflecting cuts separating any pair of vertices as long as the pair is not connected by an $A$-edge. For such pairs we define
    \begin{align}
        \CM^{(k)}_{u,v} = \CM^{(k)}_{e_u,e_v} 
    \end{align}
    where $e_u$ and $e_v$ are the $A$-edges incident on $u$ and $v$ respectively. This is a positive definite matrix. 
    When $u$ and $v$ are connected by $A$-edge, consider a different type of edge, say $B$. The $B$-edge incident on $u$ and $v$ are then distinct. Let $\tilde k$ be the reflecting cut separating these $B$-edges. It must cut the $A$-edge joining $u$ and $v$. Because of the reflecting cut property, $u$ and $v$ are images of each other under the reflecting cut $k^*$. This can be done for every pair of $u$ and $v$ that is connected by $A$-edge. Let the associated reflecting cut be $\tilde k_{V(u,v)}$. Then for such pairs we take
    \begin{align}
        \CM^{(\tilde k_{V(u,v)})}_{u',v'}  = \delta_{u,u'}\delta_{v,v'}.
    \end{align}
    This is also a positive definite matrix (a single entry of $1$ on the diagonal). This shows that the graph is also vertex-convex. 
\end{proof}

\bibliography{LargeDCFT}

\end{document}